\documentclass{tPHM2e}
\usepackage{amsmath}
\usepackage{amsfonts}
\usepackage{graphicx}
\usepackage{times}
\usepackage{color}
\usepackage{mhequ}
\usepackage[scanall]{psfrag}
\usepackage{epsfig}

\def\thecomma{\ifx,\thenext \else\ifx;\thenext \else\ifx.\thenext \else\ifx!\thenext \else\ifx:\thenext \else \  \fi\fi\fi\fi\fi}
\def\condblank{\futurelet\thenext\thecomma}
\def\ie{{\it i.e.}\condblank}
\def\eg{{\it e.g.}\condblank}
\def\estimatename{\scshape\bf Estimate}
\newtheorem{estimate}{\estimatename}[section]

\def\fref#1{Fig.~\ref{#1}}
\def\tabref#1{Table~\ref{#1}}
\def\cref#1{Condition~\ref{#1}}
\def\eref#1{(\ref{#1})}

\def\pref#1{Proposition~\ref{#1}}

\def\OO{\mathcal{O}}

\def\TT{\mathcal{T}}

\def\eref#1{Eq.~(\ref{#1})}
\def\fref#1{Fig.~\ref{#1}}
\def\sref#1{Sect.~\ref{#1}}
\let\epsilon=\varepsilon
\let\rho=\varrho
\def\nr{n_{\rm r}}
\def\nb{n_{\rm b}}
\def\rr{{\rm rr}}
\def\rb{{\rm rb}}
\def\bb{{\rm bb}}

\def\dE{{\rm d}E}
\def\0{{\tt0}}
\def\+{{\tt+}}
\def\m {{\tt-}\condblank}
\def\mm {{\tt--}\condblank}

\let\theta=\vartheta

\begin{document}

\title{Decay of Correlations in a Topological Glass}
\author{Jean-Pierre Eckmann${}^{{\rm a,b}}$ and Maher Younan
${}^{\rm a}$\\
${}^{\rm a}${{\em D\'epartement de Physique Th\'eorique, Universit\'e de
Gen\`eve;\\${}^{\rm b}$Section de Math\'ematiques, Universit\'e de
Gen\`eve}}\\
\vspace{6pt}\received{Received 00 Month 2011; final version received
  00 Month 2011} }
\maketitle
\thispagestyle{empty}

\centerline{Dedicated to David Sherrington, with admiration and best wishes}
\vspace{6pt}
\begin{abstract}
In this paper we continue the study of a topological glassy system. 
The state space of the model
is given by all triangulations of a sphere with $N$ nodes, half of which
are red and half are 
blue. Red nodes want to have 5 neighbors while blue ones want
7. Energies of nodes with 
other numbers of neighbors are supposed to be positive. The dynamics
is that of flipping 
the diagonal between two adjacent triangles, with a temperature dependent
probability. We consider the system at very low temperatures.

We concentrate on several new aspects of this model: Starting from a
detailed description of the stationary state, we conclude that pairs
of defects
(nodes with the ``wrong'' degree) move with very high mobility along
1-dimensional paths.  As they
wander around, they encounter single defects, which they then move
``sideways'' with a geometrically defined probability. This induces
a diffusive motion of the single defects. If they meet, they
annihilate, lowering the energy of the system. We both estimate the
decay of energy to equilibrium, as well as the correlations. In
particular, we find a decay like $t^{-0.4}$.

\end{abstract}

\section{Introduction, the Model}\label{s:model}
This paper deals with a species of a class of models on topological
studies of triangulations. Such models have been studied in several
contexts 2-d gravitation, froth, \cite[and references
    therein]{schliecker2002}. The 
variant we use here was introduced in \cite{eckmannglass2007}, but it turned out
that a very similar study was initiated earlier by Aste and
Sherrington \cite{astesherrington1999}. So, we hope that David will accept this paper
as a small sign of recognition. 

We reconsider here the model which was inspired by \cite{Proglass2007}
and introduced in \cite{eckmannglass2007}. For
completeness we repeat the definition of the model: We fix a (large)
number $N$ of nodes, half of which are red, and the other half blue. These
nodes are the nodes of a topological triangulation $T$ of the sphere
$S^2$. The set of all possible such labelled triangulations will be denoted $\TT_N$.
We define a dynamics on $\TT_N$ by the following 
Metropolis algorithm whose elementary steps are
flips (T1 moves): A link is chosen uniformly at random (among the
$3N-6$ links). In \fref{f:abcd}, if the link AB was chosen
then the flip consists in replacing it
by the link CD. This move is not admissible if the link CD already exists
before the move. Otherwise it is admissible. Note that the number of
nodes, $N$, does not change in this model. However, we will be
interested in the behavior for $N\to\infty$.

The Metropolis algorithm is based on the
energy function $E$ on $\TT_N$ which, for any triangulation
$T\in\TT_N$, is defined as 
\begin{equ}
  E(T)= \sum_{i\in \text{blue}} (d_i-7)^2 + \sum_{i\in \text{red}}
  (d_i-5)^2~,
\end{equ}
where $d_i$ is the degree (number of links) of the node $i$. Thus,
this energy favors 7 links for the blue nodes and 5 for the red
ones. {\it Mutatis mutandis}, the detailed definition of the energy is
not important for the discussion of the model, and we will stick to
this particular form of the energy. Given an admissible flip, compute
the energy of the triangulation before and after the flip; this
defines 
\begin{equ}
  \dE = E_{\text{after}} -E_{\text{before}}~.
\end{equ}
An admissible flip is performed if
either $\dE \le 0$ or, when $\dE>0$, with probability $\exp (-\beta\dE)$,
where $\beta $ is the inverse temperature of the system.

Several properties of this model were discussed in \cite{eckmannglass2007}, but here
we study in more detail the dynamical properties of the model. In
particular, we introduce a ``charge''
defined as follows:
\begin{definition}
  The \emph{charge} of a red node is defined by $d_i-5$ and the charge of a
  blue node is defined by $d_i-7$. We will say the charge is a
  \emph{defect} \+ if it
  is $+1$ and \m if it is $-1$. In general, the color of the charge
  will not matter and will not be mentioned. 
\end{definition}
In principle, all charges
  between $-4$ and $\OO(N)$ can occur, since $d_i\ge3$, but, obviously, at low
  temperatures mostly the charges \+, \0, and \m will come into play.

\section{Equilibrium and the Approximation of the Dynamics}\label{s:approximation}
The dynamics of the model is given by the Metropolis algorithm. In it,
a link is chosen uniformly at random among all possible links. The
change of energy induced by the  flipping of this link is called
$\dE$. If $\dE\le 0$ the flip is performed, if $\dE>0$ the flip is
performed with probability
$p(\dE)=\exp(-\beta \dE)$. This process satisfies detailed balance, and most
of the paper is dealing with the equilibrium properties of this
process at low temperatures. Because of the detailed balance, the equilibrium measure
$\mu$ has the property that the probability to see a given state whose
energy is $E$ is proportional to $\exp(-\beta E)$.  We use this
elementary observation to argue that at low temperature there are only few
defects, by which we mean that there are few red nodes whose degree is
not 5 and also few blue ones whose degree is not 7. Given that there
are few of these ``defects'', we further assume that the ``positions''
of these defects are random in the sense that there are no strong
conditional expectations: For example, having a defect +1 does not say
that there is a defect -1 close-by. The upshot of this way of
reasoning, which we corroborate by numerical studies, is that one can
approximate the dynamics by just looking at defects.

Indeed, the full dynamics must be described by the evolution of
correlation functions.
It would have to take into account correlation
functions between the charges (and the colors) of, say, the 4 nodes on
a pair of triangles sharing an edge. Then, flipping that edge, the
correlations of many neighboring triangles would be changed
simultaneously, and this would necessitate considering a full hierarchy of
correlations (like BBGKY).  
What we will see is that in this model, these higher
order correlation functions do not influence our
basic understanding of what is going on. 

In contrast, the Euler relations play a small but not totally
negligible role for the sizes of the systems we consider.

\section{Description of the Stationary State}

It will be useful to define throughout the paper the natural parameter
\begin{equ}
  \epsilon \equiv e^{-\beta }~.
\end{equ}
We are interested in a regime where the density $c$ of charges (which
equals $E/N$) is low
but also, where the number $c\cdot N$ of charges is large, so that good
statistics and a certain independence of the Euler relations is
attained.
More precisely, we fix $\rho\ll 1$ and $D_0\gg 1$, and require
$\epsilon \le \rho$ and $N\epsilon >D_0$. We furthermore consider the limit of
large $N$.

The main result of this section is summarized in the following proposition:
\begin{proposition}
Consider an equilibrium state at temperature $T\ll1$ satisfying the
above conditions on $N$ and $\epsilon$.
 \begin{enumerate}
  \item At first order in $\epsilon$, the only charges present in the
    system are   simple defects $\pm 1$. Their density is
    $2\epsilon+\OO(\epsilon ^2)$.
  \item The distribution of the colors (red or blue) is independent in
    the limit $\epsilon \to0$. 
  \item The distribution of the charges is independent in the limit
    $\epsilon \to 0$.
 \end{enumerate}

 \begin{remark}
   The meaning of $\epsilon \to 0$ above is that the quantities become
   more and more decorrelated as $\epsilon \to0$ while still
   maintaining the inequalities $\epsilon \le \rho$ and $N\epsilon >D_0$. 
 \end{remark}

\end{proposition}

\subsection{Energy of the stationary state}
In this paragraph, we will calculate the energy of the
stationary state in the limit specified above, as a function of the temperature.

\begin{estimate}\label{est:1}
Consider the region $\epsilon N>D_0$ and $\epsilon <\rho$.
For sufficiently large $D_0$ and sufficiently small $\rho $ the
density of charges $c$ is 
\begin{equ}
c\equiv E/N=2\epsilon+\OO(\epsilon ^2) ~.
\end{equ}
\end{estimate}

\begin{proof}
Assuming equilibrium, by detailed balance, the probability to see a
defect of charge $\pm1$ is 
$\OO(e^{-1\cdot\beta })=\OO(\epsilon )$, while the probability to see higher charges is
$\OO(e^{-2^2\beta })=\OO(\epsilon^4 )$, by the assumption of
equilibrium and the form of the Hamiltonian, since, if $(d_i-5)^2>1$
then it is at least 4.

So it remains to
estimate the coefficient in front of the factor $\epsilon $.
There are 4 cases to consider: The number of red nodes with degree 4
or 6, resp.~the number of blue nodes with degree 6 or 8. All these
cases cost energy 1 per instance, and thus these 4 numbers are equal by
the virial theorem.

We also need to estimate the cases with 0 charge, \ie, blue nodes with
7 neighbors and red nodes with 5 neighbors,
which appear again equally often, by
the virial theorem.
Since there are $N/2$ nodes of each color, and each of the colors has
2 states of defect 1 (namely $\pm1$), we conclude that the expected
total number of defects is
\begin{equ}\label{e:twonepsilon}
  2\cdot 2 \cdot\epsilon  \cdot (N/2) = 2\epsilon N+\OO(\epsilon ^2)~.
\end{equ}

\end{proof}

\subsection{Distribution of the colors}\label{s:colordist}

We next calculate the probabilities that a randomly 
chosen link connects 2 red (blue) nodes. We denote these probabilities
by $p_\rr$ for red-red, $p_\rb$ for red-blue and so on.
If there are no defects, \ie, at order $\epsilon ^0$, all red nodes
have $5$ neighbors and all blue nodes have $7$. This leads to the following relations:
\begin{equa}
 2p_{\rr} + p_{\rb} &= 5/6~, \\
 2p_{\bb} + p_{\rb} &= 7/6~.
\end{equa}
Assuming that the positions of the colors are uncorrelated, we find
that the relative probabilities to find a red-red, resp.~blue-blue pair are
\begin{equ}
 p_{\rr} / p_{\bb} = 25/49~.
\end{equ}
This leads to $p_{\rr} = 25/144$, $p_{\bb} = 49/144$, and $p_{\rb} = 70/144$.
In \fref{f:color-distribution} we show that numerical simulations confirm this simple approximation to a very high degree of fidelity.

\begin{figure}
 \begin{center}
    \epsfig{file=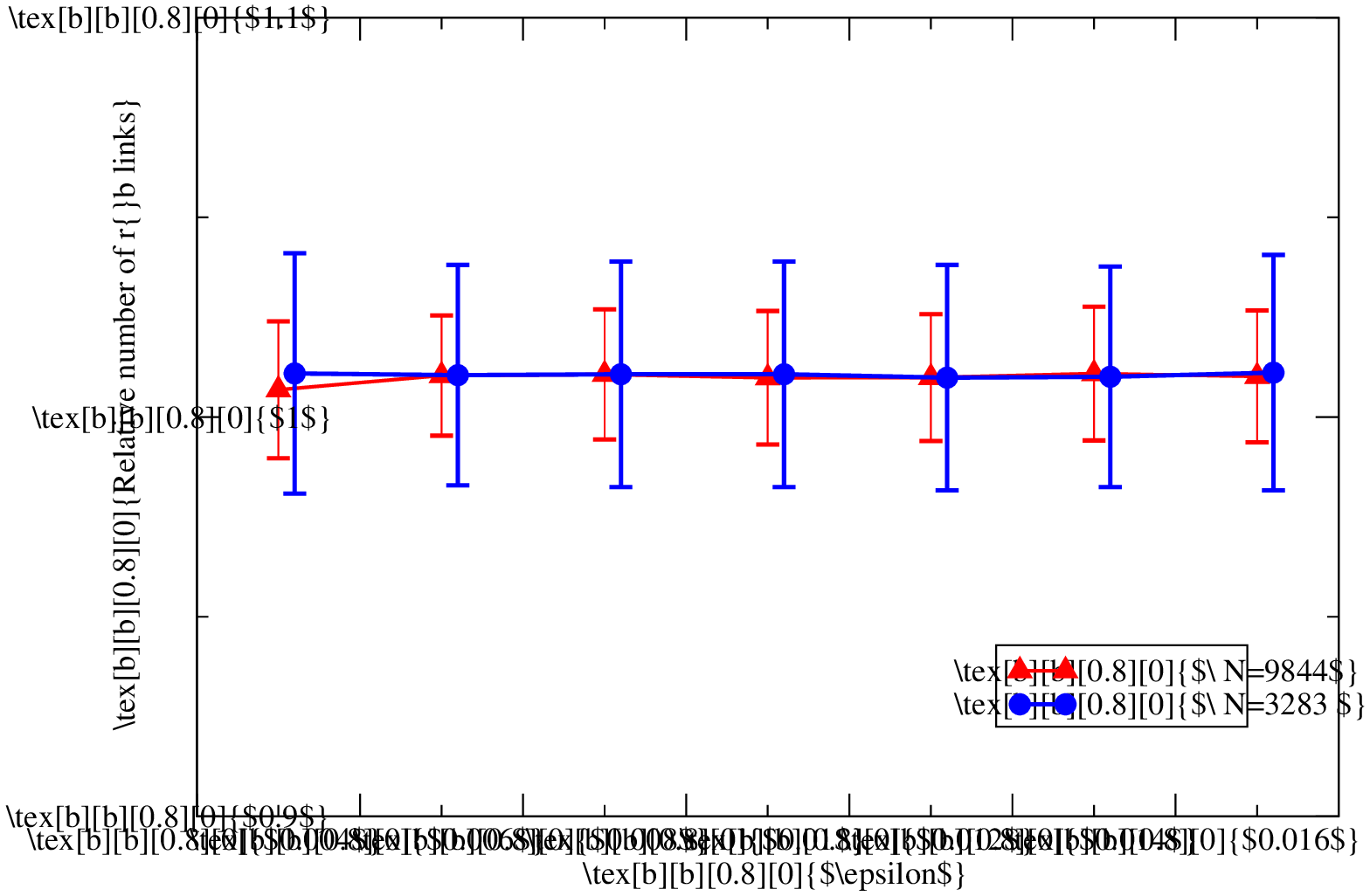,width=0.5\textwidth,angle=270}
  \end{center}  
  \caption{Numerical check of relation $p_{\rb}=70/144$ by plotting
    $p_{\rb}/(70/144)$. The error bars are 3 $\sigma$ and the data for
    $N=3283$ are slightly shifted (in the $x$-direction) for better visibility.}\label{f:color-distribution} 
  
\end{figure}

\subsection{Energy cost of flips}

We adopt an approach similar to \sref{s:colordist}. We use the hypothesis that
the charges are randomly distributed over the nodes to calculate
the probability of finding a link with a given neighborhood of charges
and compare it to simulation results.
In this case however, given a link $\ell$, the neighborhood we consider is the ordered set 
of all 4 nodes involved in its flipping. For example in \fref{f:abcd}, this set would be
$\left( c(A),c(B),c(C),c(D) \right)$ where $c(A)$ is the charge of the node $A$.
This choice will be very useful for to study the dynamics later on since it determines
the energy cost of flipping a given link:
\begin{equs}[e:energy-cost]
  \dE(\ell) &= \sum_{n\in\{A,B\}}\left( c(n)-1 \right)^{2} - \left( c(n) \right)^{2} +  \sum_{n\in\{C,D\}}\left( c(n)+1 \right)^{2} - \left( c(n) \right)^{2} \\
	    &= 4+2\left( c(C)+c(D)-c(A)-c(B) \right)~\text{.}
\end{equs}

It is easy to enumerate all the various cases and the energy cost
associated with each of them.
We restrict the discussion to those situations where the charges take
values in $\{+1,0,-1\}$.
\begin{figure}
  \begin{center}
    \epsfig{file=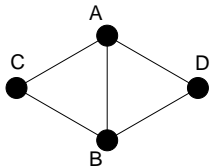,width=0.2\textwidth}
  \end{center}
  \caption{Labeling of the corners of 2 adjacent triangles}\label{f:abcd}
  \end{figure}
In principle, there are $3^4$ configurations, which are reduced
to $36$, by symmetry. They are summarized in Table~\ref{table:1} (symmetrical cases
omitted). 

Note that if the defects of the original configuration are bounded by
$\pm1$, then $\dE$ varies between $-4$ and $12$.
\begin{table}
\tbl{The energy differences obtained by flipping the link between
the first 2 values to a link between the second 2 values, as a
function of the number of defects.}
{\begin{tabular}{c c c c c r}
\toprule
defects & \multicolumn{4}{c}{initial state} & $\dE$\\
\colrule
0& 0 & 0 & 0 & 0 &  4\\
\colrule
1& + & 0 & 0 & 0 &  2\\
1& 0 & 0 & 0 & - &  2\\
1& 0 & - & 0 & 0 &  6\\
1& 0 & 0 & + & 0 &  6\\
\hline
2& + & + & 0 & 0 &  0\\
2& + & 0 & 0 & - &  0\\
2& 0 & 0 & - & - &  0\\
2& + & - & 0 & 0 &  4\\
2& + & 0 & + & 0 &  4\\
2& 0 & - & 0 & - &  4\\
2& 0 & 0 & + & - &  4\\
2& - & - & 0 & 0 &  8\\
2& 0 & - & + & 0 &  8\\
2& 0 & 0 & + & + &  8\\
\botrule
\end{tabular}\quad
\begin{tabular}{c c c c c r}
\toprule
defects & \multicolumn{4}{c}{initial state} & $\dE$\\
\colrule
3& - & - & + & 0 &  10\\
3& 0 & - & + & + &  10\\
3& + & - & 0 & - &  2\\
3& + & + & + & 0 &  2\\
3& + & + & 0 & - &  -2\\
3& + & 0 & - & - &  -2\\
3& + & 0 & + & - &  2\\
3& 0 & - & - & - &  2\\
3& - & - & 0 & - &  6\\
3& + & - & + & 0 &  6\\
3& + & 0 & + & + &  6\\
3& 0 & - & + & - &  6\\
\botrule
\end{tabular}\quad
\begin{tabular}{c c c c c r}
\toprule
defects & \multicolumn{4}{c}{initial state} & $\dE$\\
\colrule
4& + & - & - & - &  0\\
4& + & + & + & - &  0\\
4& - & - & + & + &  12\\
4& - & - & - & - &  4\\
4& + & - & + & - &  4\\
4& + & + & - & - &  -4\\
4& + & + & + & + &  4\\
4& - & - & + & - &  8\\
4& + & - & + & + &  8\\
\botrule
\end{tabular}}\label{table:1}
\end{table}

\subsection{The number of local defect configurations}

We assume throughout that the number of red (blue) nodes is $\nr$
($\nb$) and that $\Delta\equiv \nr-\nb\in\{0,1\}$. We denote by
$p_\pm$ the probabilities to find charges $\pm1$, respectively.
Assuming that there are no other charges (except 0), we can write
\begin{equa}
N\cdot( p_{\tt -}  + p_\+) &= E ~,  \\
N\cdot( p_{\tt -} - p_\+) &= 12 - \Delta~,
\end{equa}
where the second equation follows from the Euler formula.
In equilibrium, $E=2N\epsilon $, by \eref{e:twonepsilon}, and
therefore we get
\begin{equ}\label{e:pplusminus}
 p_{\tt\pm}= \epsilon \mp 6/N \pm \Delta/(2N) + \OO\left( \epsilon^2 \right)~.
\end{equ}
We will assume that $N\epsilon\gg 6$ so that the second term in
\eref{e:pplusminus} can be neglected.
In a similar way, one can show that
\begin{equ}
  p_{\pm2}=\epsilon ^4 +\OO(\epsilon ^5)~,
\end{equ}
and combining these we find that the probability of nodes with charge
$0$ is
\begin{equ}
  p_0= 1-2\epsilon +\OO(\epsilon ^2)~.
\end{equ}
We next consider in more detail what happens in those pairs of
triangles where a flip leads to $\dE=0$. Looking again at \eref{e:energy-cost}
we see that the case $\dE=0$ appears in 3 cases:\hfill\break
Case $q_{\+\m }$: One of A or B has charge \+ and C
or D has charge \m (and the others, charge 0).\hfill\break
Case $q_{\+\+}$:  A and B charge +, C and D charge 0.\hfill\break
Case $q_{\mm }$: A and B charge 0, C and D charge \m\,.\hfill\break

Continuing with the independence assumption, we now look at the
probability to find a configuration of type $q_{\+\+}$, $q_{\+\m }$, and
$q_{\m {}\m }$. Note that there are $6N-12$ half-links emanating from the
nodes, and we are to pair them up randomly. Note that if a site is red, it has $4$, $5$, $6$ outgoing links,
depending on whether its charge is $-$, $0$, $+$, respectively.
Similarly, the numbers for a blue node are $6$, $7$, $8$.
Therefore, given that there are on average $\epsilon N/2$ defects of
type red-$4$, red-$6$, blue-$6$, blue-$8$, there will be
$4\epsilon N/2$ links from the red-$4$, $6\epsilon N/2$ from red-$6$
and blue-$6$, and $8\epsilon N/2$ from blue-$8$.
The blue-$7$ and red-$5$ occur with probability almost 1 and have
therefore
respectively $7 N/2$ and
$5 N/2$ dangling edges (with a correction factor
$1-\OO(\epsilon ))$ which we omit throughout.
The probabilities to see such dangling edges are the quantities above,
divided by $6N-12$, the total number of dangling edges. We get,
omitting higher order terms:
\begin{equa}[e:qpm]
 q_{\+\+} =& \left( 7p_{+}/6 \right)^{2} \cdot p_0^{2}=49\epsilon ^2/36~, \\
 q_{\mm } =& \left( 5p_{-}/6 \right)^{2} \cdot p_0^{2}=25\epsilon ^2/36~, \\
 q_{\+\m } =& 4\left( 5p_{-}/6 \right)\left( 7p_{+}/6 \right) \cdot p_0^{2}=140\epsilon ^2/36~. \\
\end{equa}
We also get, by looking at \tabref{table:1}:
\begin{equa}[e:pdE]
   p_{\dE=0} &= q_{\+\+} + q_{\mm } + q_{\+\m }=214\epsilon ^2/36~,\\
   p_{\dE=2} &= 2\left( 7p_{+}/6 + 5p_{-}/6 \right) \cdot  p_{0}
   ^{3}=4\epsilon  ~,\\
   p_{\dE=4} &= p_0^4=1-\OO(\epsilon )~,
\end{equa}

The discussion of the other values of $\dE$ shows the limitations due
to our closing assumptions: by the virial theorem, in total
independence, we would simply have
\begin{equ}\label{e:pdE2}
p_{\dE=0}=p_{\dE=8}\text{ and }p_{\dE=2}=p_{\dE=6}~.
 \end{equ}
But we could also have computed the probabilities as above, with the result:
\begin{equa}\label{e:pdE3}
 p_{\dE=-2} &= 2\left( 7p_{+}/6 \right)^{2} \cdot \left( 5p_{-}/6 \right) \cdot p_0 + 2\left( 5p_{-}/6 \right)^{2} \cdot \left( 7p_{+}/6 \right) \cdot p_0 \\
	   &\approx 3.89 \epsilon ^{3}
\end{equa}
instead of $4\epsilon^{3}=p_{\dE=2}\cdot\epsilon^{2}$ given by the 
stationarity 
assumption, which proves that the distribution of
defects is not completely uncorrelated. We will say that the
correlation is bounded by $0.1 \epsilon ^3$, and can thus be neglected
in the limit $\epsilon \to0$.

In Figs.~\ref{f:energy} and \ref{f:fitde0} we show with 2 examples that the numerical simulations confirm
these simple approximations to a very high degree of fidelity.
Note that in \cite{Godreche1992}, the
\emph{uniform} measure on $\TT_N$ was considered, and even this leads to
correlations of degrees of neighboring nodes.

\begin{figure}
  \begin{center}
    \psfig{file=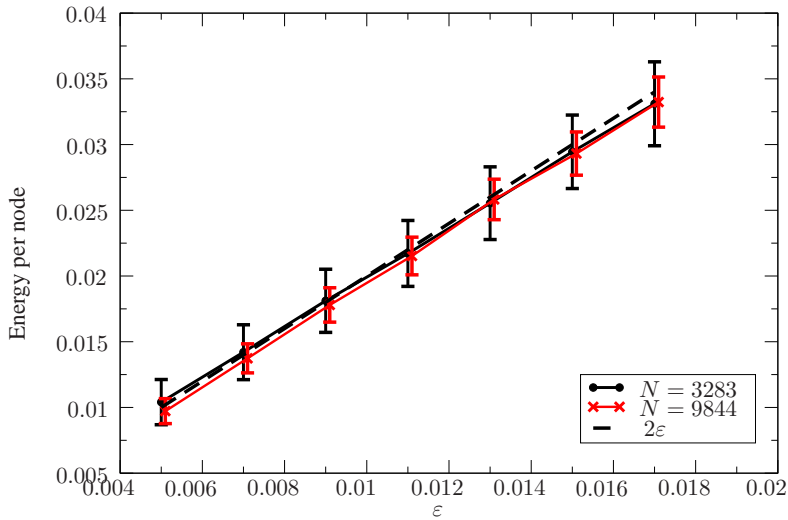,width=0.5\textwidth,angle=270}
  \end{center}
  \caption{Numerical test of the mean energy  per node (Estimate
    \ref{est:1}) for 950 realizations.
    The data for $N=9844$ are
    slightly shifted for better visibility. Note the excellent fit
    with the theoretical curve, although the fluctuations are huge,
    getting better with larger system size (1 standard
    deviation shown).}\label{f:energy}
\end{figure}
\begin{figure}
  \begin{center}
    \psfig{file=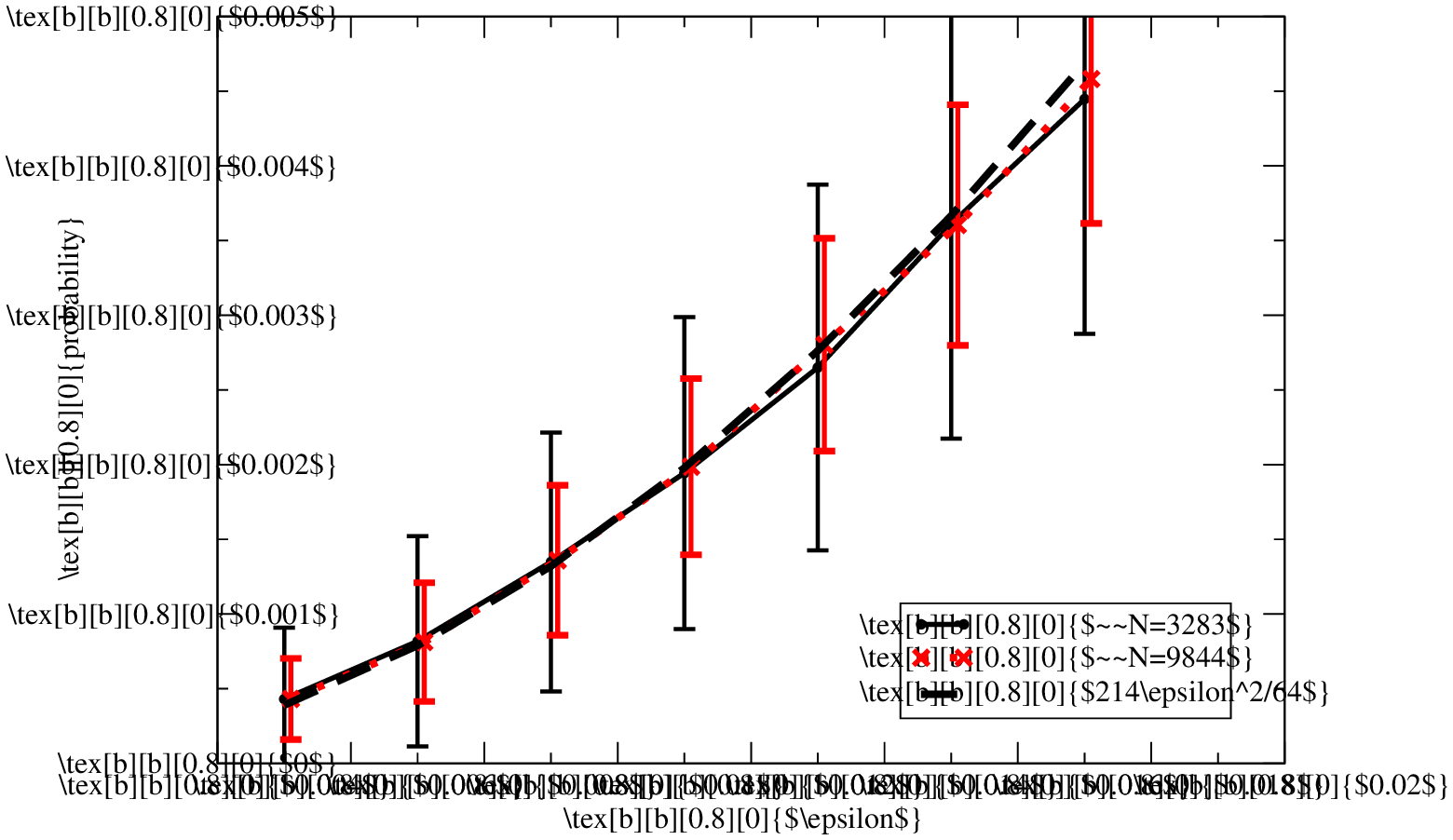,width=0.5\textwidth,angle=270}
  \end{center}
  \caption{Numerical test of $p_{\dE=0}$ (\eref{e:pdE})  from 950 realizations. The data for $N=9844$ are
    slightly shifted for better visibility. Note the excellent fit
    with the theoretical curve, although the fluctuations are huge,
    getting better with larger system size (1 standard
    deviation shown).}\label{f:fitde0}
\end{figure}

\section{Dynamics of the System (at Equilibrium)}
\label{s:dynamics-eq}
We can use the results of the previous section to
estimate the dynamics of the system under the Metropolis algorithm.

If a flip leads to an energy change $\dE$ then it is accepted in the
Metropolis algorithm with probability
\begin{equ}\label{e:accept}
  p_{\rm{acceptance}}= \epsilon ^{\max(0,\,\dE)}~.
\end{equ}
On the other hand, the probabilities to pick a link with fixed
$\dE$ are given by \eref{e:pdE} and \eref{e:pdE2}.
Multiplying these numbers with the probabilities of \eref{e:accept}
leads to an estimate of the probability that the flip in question actually happens.
The results are summarized in \tabref{table:2} (calculated this time
with the method of \eref{e:pdE3}).

\begin{table}
\tbl{The probabilities of transitions from the initial state. Data
only shown to order $\epsilon ^4$. The third column shows the
probabilities to pick a link leading to a given $\dE$. Higher order
corrections are omitted.}
{\begin{tabular}{r r r}
\toprule
$\dE$ & \qquad\qquad transition rate & \qquad \qquad local landscape\\
\colrule
-4 & $ 1225/1296\cdot \epsilon^{4}$ &  $ 1225/1296\cdot \epsilon^{4}$ \\
-2 & $ 35/9\cdot \epsilon^{3}$ &  $ 35/9\cdot \epsilon^{3}$ \\
0 & $ 107/18\cdot \epsilon^{2}$ &  $ 107/18\cdot \epsilon^{2}$ \\
2 & $ 4\cdot \epsilon^{3}$ &  $ 4\cdot \epsilon^{1}$ \\
4 & $ 1\cdot \epsilon^{4}$ &  $ 1\cdot \epsilon^{0}$ \\
6 & $ 4\cdot \epsilon^{7}$ &  $ 4\cdot \epsilon^{1}$ \\
8 & $ 107/18\cdot \epsilon^{10}$ &  $ 107/18\cdot \epsilon^{2}$ \\
10 & $ 35/9\cdot \epsilon^{13}$ &  $ 35/9\cdot \epsilon^{3}$ \\
12 & $ 1225/1296\cdot \epsilon^{16}$ &  $ 1225/1296\cdot \epsilon^{4}$\\
\botrule
\end{tabular}}
\label{table:2}
\end{table}
{\bf Discussion}: Inspection of \tabref{table:2} shows that the 
events with the highest transition rate are those which cost no energy, followed by those which 
have an energy cost of $\pm 2$.
Also note that the
probability to find a link which will lead to a given $\dE$ is equal
to the quantity in the table times $\epsilon ^{-\max(0,\,\dE)}$ since
then we neglect the Metropolis factor. This leads to a table with the
same prefactors, but with a power $\epsilon ^{|\dE-4|/2}$.
In particular, \emph{in the steady state, the local landscape is given
by the 3rd column of Table~\ref{table:2}: It is symmetric around $\dE=4$.}\newline

Henceforth, we will only consider the 3 most frequent types of flips
(the others are an order $\epsilon $ less probable):
\begin{enumerate}
 \item Flips which change from 1 defect to 3 of them
and which raise the energy by 2. 
 These flips will be called \emph{creation events}.
 \item Flips which start from 3 defects and end with 1 defect and which decrease the energy by 2.
 These flips will be called \emph{annihilation events}. Creation and annihilation
 events are obviously dual to each other and equiprobable
 in the stationary state.
 \item Flips which do not change the energy, and in which a pair \+\+,
  \+\m, or \mm is involved.
 These flips are by far the most probable. 
 We will discuss below in more details
 the 3 configurations which lead to $\dE=0$.
\end{enumerate}

\subsection{The most probable flips}

As stated above, if $\epsilon=1\%$, then over $99\%$ of the flips
(which are accepted by the Metropolis algorithm)
do not change the energy. It is clear that, in order to understand
the dynamics of the system, one should start by studying these flips.

Looking at \tabref{table:1} we see that there are 3 candidates for
$\dE=0$ and they all involve only 2 defects. We will now show that the cases
of \+\+\0\0 and \0\0\mm are quite different from that of \+\0\0\m
(and its 3 other variants \+\0\texttt{-}\0,\,\,$\dots$\,). In the first case, \+\+\0\0, which is 
similar to the case \0\0\mm, the local
neighborhood looks like in \fref{f:backandforth}.
\begin{figure}
\begin{center}
\psfig{file=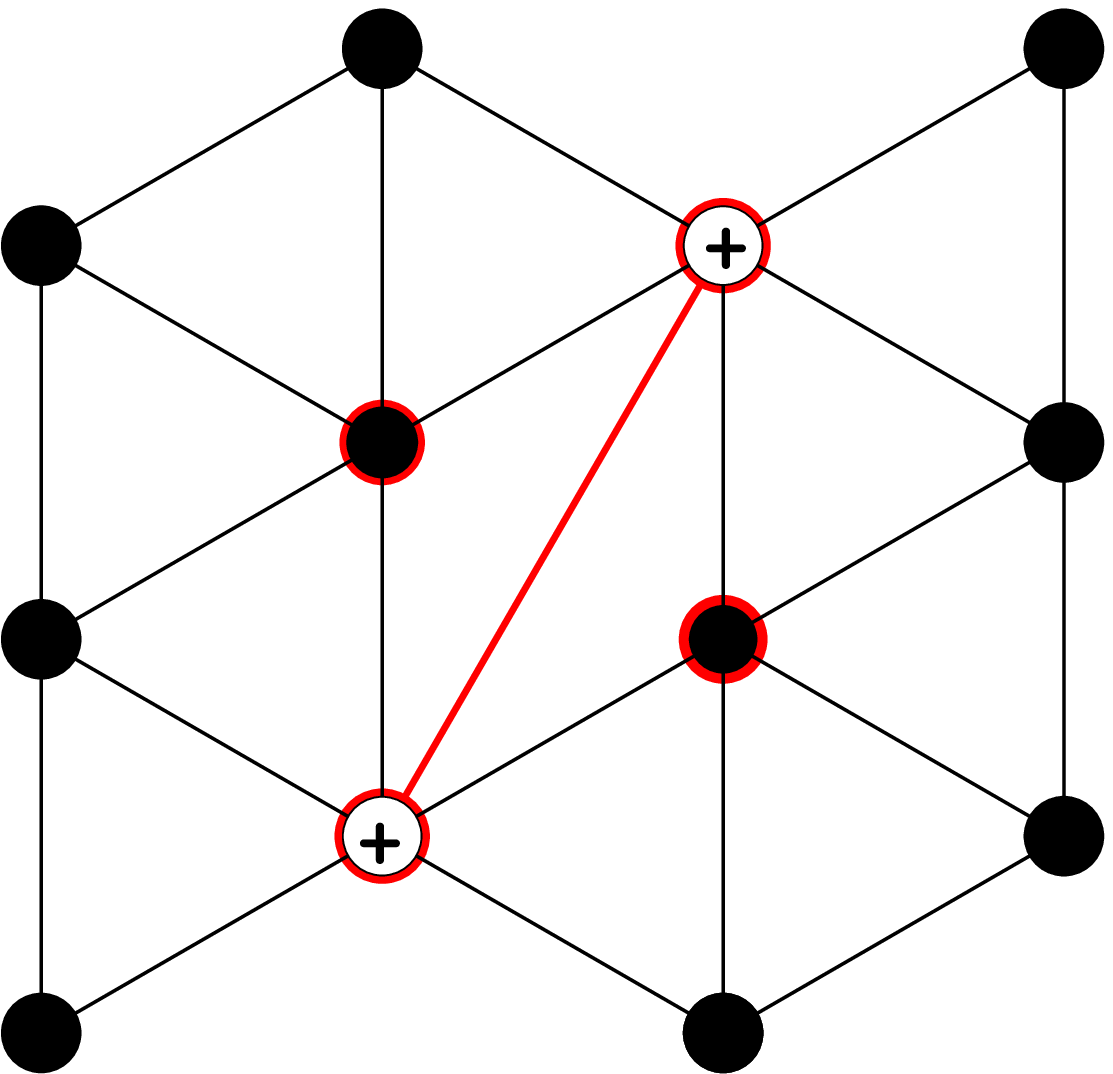,width=0.3\textwidth}\qquad\psfig{file=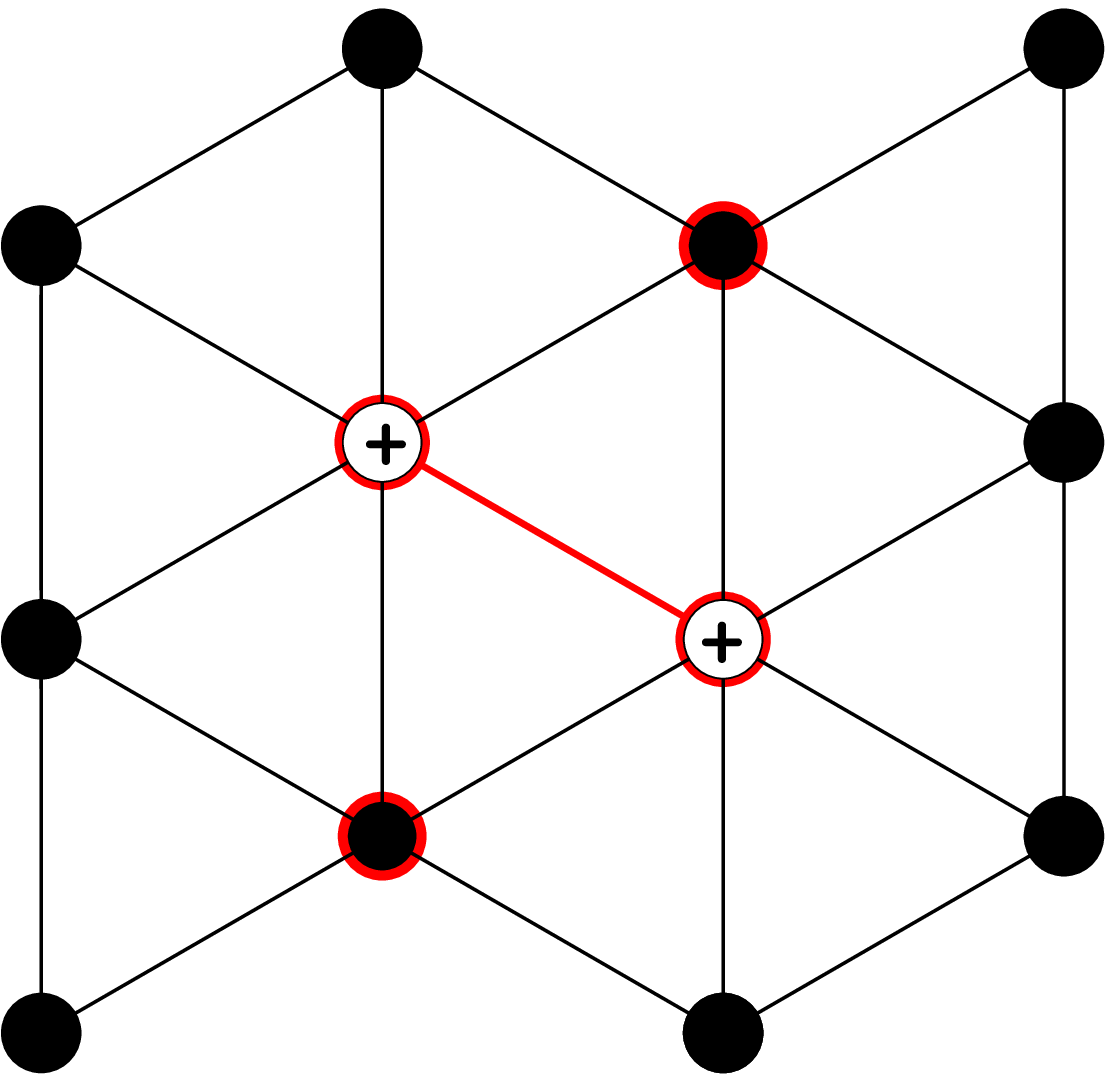,width=0.3\textwidth}
\end{center}
  \caption{A flip from \+\+\0\0 (on the left) and the resulting
    triangulation on the right. The affected nodes are supposed to be
    red, in this example. Note that the result is \emph {again} of the
    type \+\+\0\0. Furthermore, \emph{again} with $\dE=0$ one can flip
    back. This is reminiscent of ``blinkers'' in the game of life \cite[Chap25]{BCG}.}\label{f:backandforth}
\end{figure}
In this case, what happens is a flipping back and forth between the 2
states, with probability $p=(3N-6)^{-1}$ (the probability to choose the
colored link).

The case \+\0\0\m  is illustrated in \fref{f:goaway}. Here a new, and
important phenomenon appears: The pattern, \+\0\0\m which we will call a \emph{pair}, is recreated, but
\emph{at a new position} a distance 1 away from the old one. We will
also say that the pair \+\m \emph{walks} one step.
\begin{figure}
\psfig{file=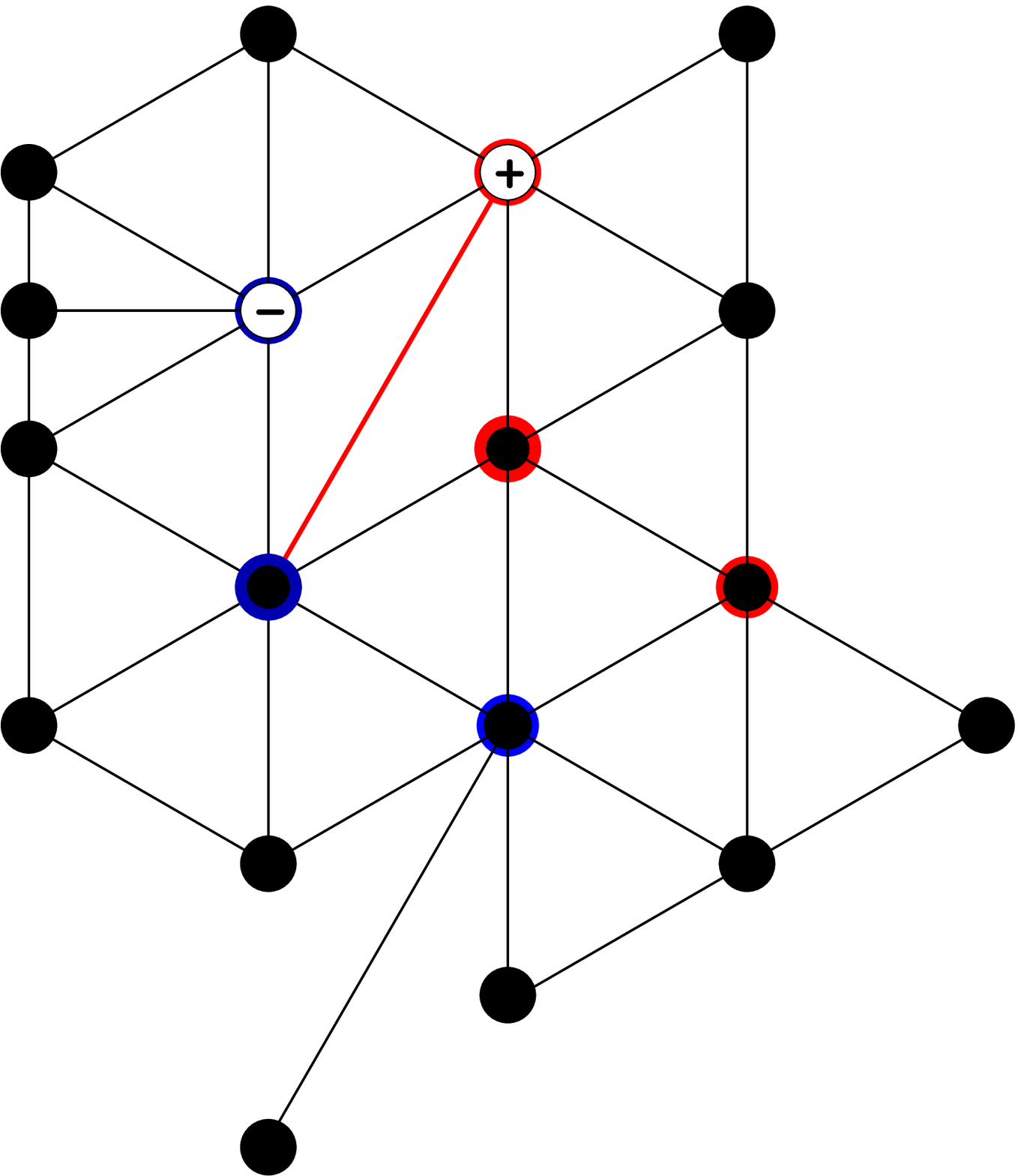,width=0.3\textwidth}\qquad\psfig{file=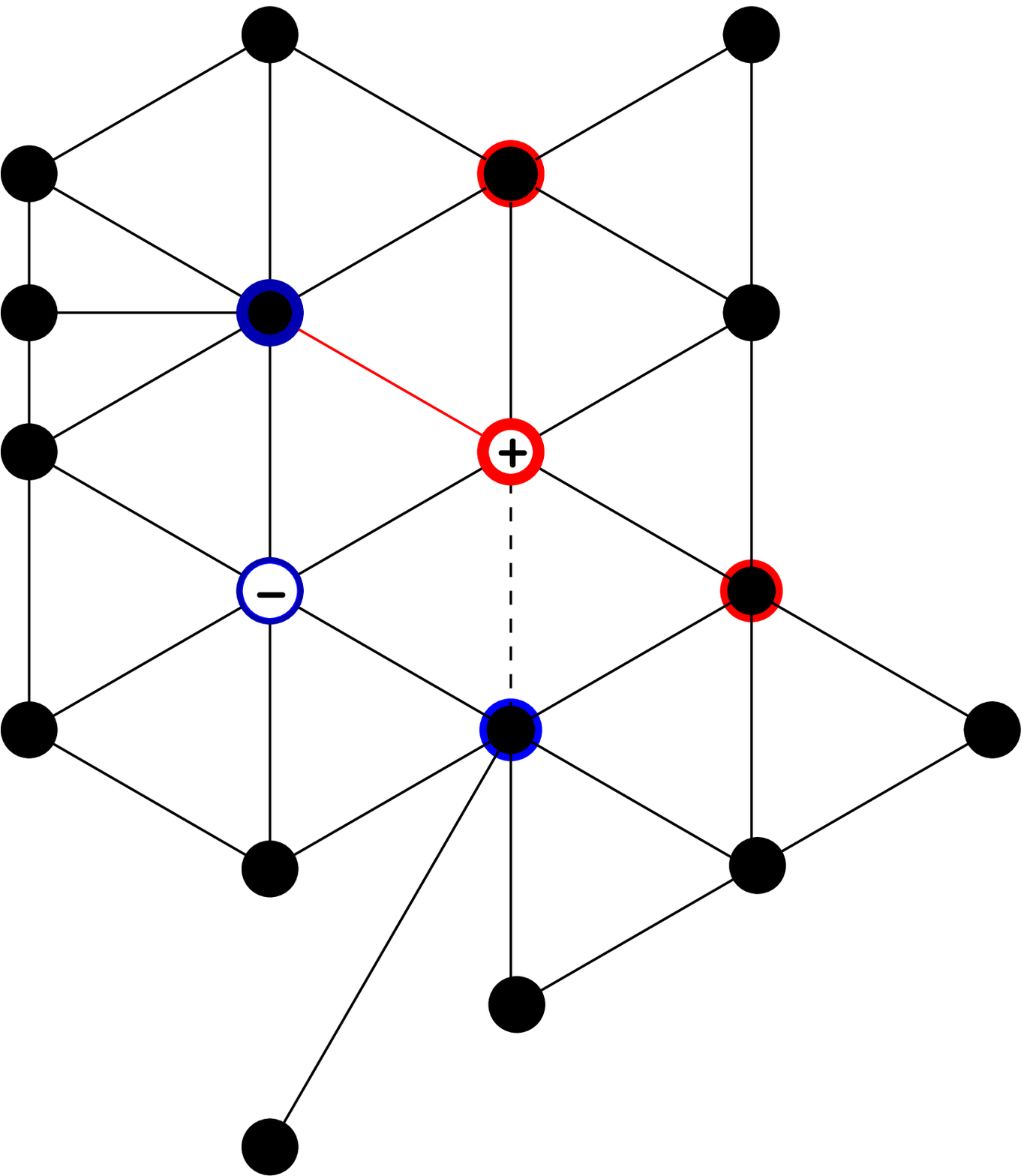,width=0.3\textwidth}\qquad\psfig{file=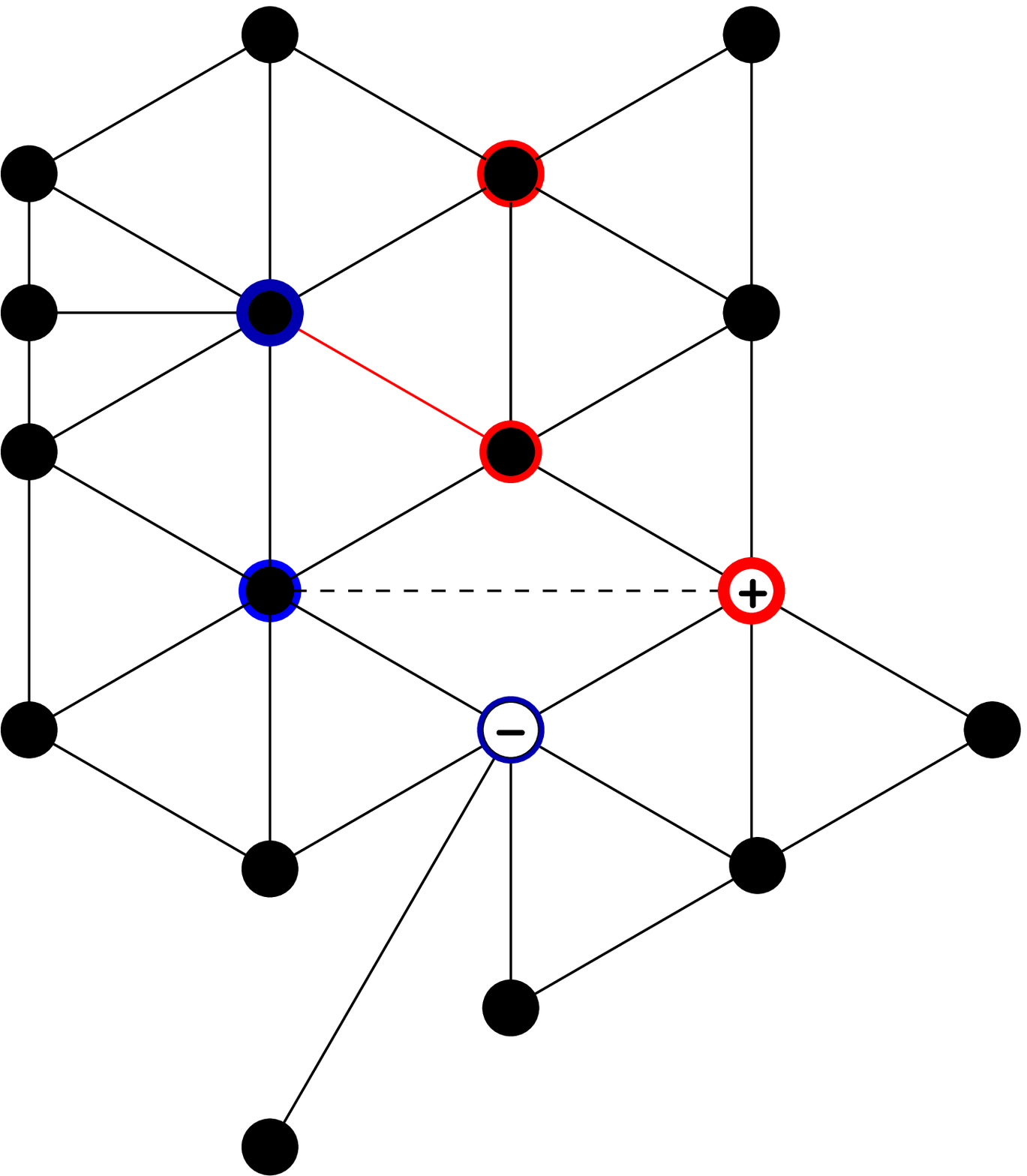,width=0.3\textwidth}
  \caption{Change of pattern in the case \+\0\0\m . In the sample only the
    relevant colors are as shown. Note that the effect of the flip is
    that the 2 defects move (in the picture) \emph{down}. The
    reverse flip costs nothing $\dE=0$. The second flip (dashed line)
  moves the defects one step further. Note that this motion must take
  place one a \emph{predefined, 1-dimensional path}.}\label{f:goaway}
\end{figure}

The more important observation is that the pairs of defects must walk
along a \emph{predefined, 1-dimensional path} as shown in
\fref{f:path}. This means that
\emph{the $dE=0$ motion of \+\m pairs is a {\bf{one}}-dimensional
  random walk in the current triangulation $T$. This random walk
  (flipping back and forth on the predefined path) will continue until
  some other type of event happens.}

\begin{figure}
\begin{center}
\psfig{file=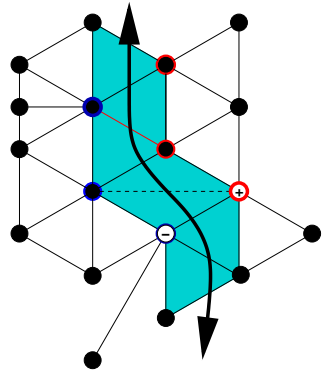,width=0.3\textwidth}
\end{center}
  \caption{The same configuration as in \fref{f:goaway} with the \emph{1-dimensional}
     path of the pair of defects superimposed.}\label{f:path}
\end{figure}

\subsection{Lifetime of pairs}\label{s:lifetime}
As we have seen, a pair of opposite charges \+\m   can move through the
system without energy cost. Its motion is a 1D random walk along a
fixed 1-dimensional path. Edges are still chosen randomly and will be
flipped if possible and if the Metropolis condition is met in case
$\dE>0$. Here we ask about the relative probabilities that a pair
disappears, and we will show that \emph{predominantly a pair will die when
it collides with a defect}.

We need to compare 3 possibilities of which the first will be seen to
be the most probable:
\begin{enumerate}
  \item The random walk reaches another defect.
  \item The pair is destroyed because a creation event involving 1 or
    2 of its 2 
  defects occurs. 
  \item Two independent random walks meet.
\end{enumerate}

Our earlier discussion says that the concentration of pairs $\left( 70\epsilon^{2}/36 \right)$
is much smaller than the concentration of defects $\left( 2\epsilon \right)$, implying
that the probability of 2 pairs meeting is insignificant when compared 
to the probability of a pair meeting a defect.

We next estimate the probability of destroying a pair as in case (2).
On average, there are
7 links in the neighborhood of a given pair which increase $E$, and flipping such a link
has an energy cost of $2$.
The probability of this to happen
is $7\epsilon^{2}/(3N)$.
Since the pair moves every $\OO(N)$ attempted flips, we conclude that, on average, a pair will do
$\OO\left( \epsilon^{-2} \right)$
steps before it is destroyed as in case (2).

The number of steps needed for case (1) to happen depends obviously on the
density of defects.
We let $\xi$ denote the average distance between defects
(counted in number of links).
Since the number of defects is $2\epsilon N$ and the system is
2-dimensional,  we conclude that
\begin{equ}
  \xi =\OO(\epsilon ^{-1/2})~.
\end{equ}

As long as the pair is not destroyed by the mechanism leading to
case (2) it can thus do $\OO(\epsilon ^{-2})$ steps by which
time it can visit $\OO(\epsilon ^{-1})$ defects.

This terminates the comparison of the 3 cases, and shows
that a pair has the time to visit a very large number of 
defects before it is destroyed by the 2 other mechanisms.

\begin{figure}
  \begin{center}
    \epsfig{file=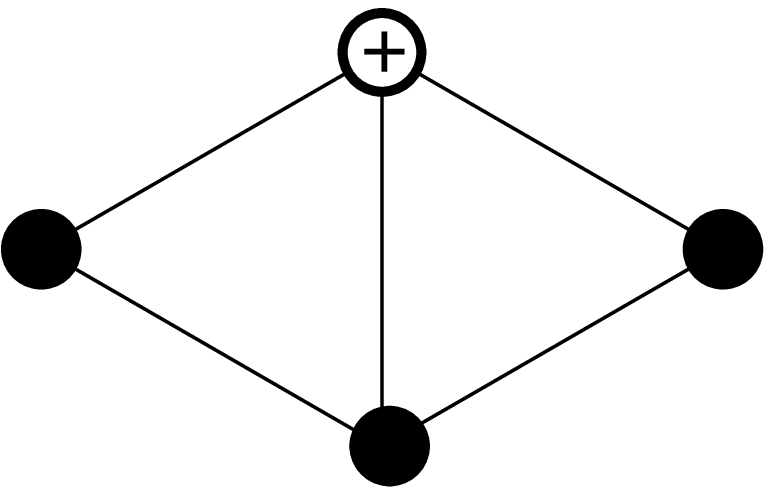,width=0.2\textwidth}\qquad
    \epsfig{file=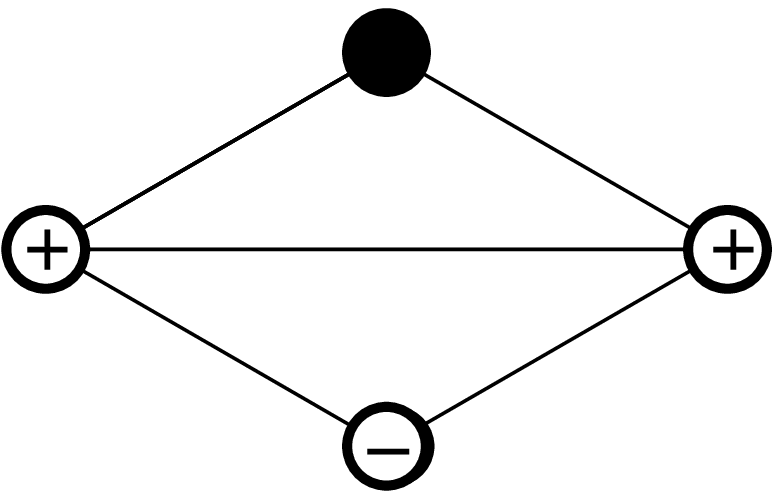,width=0.2\textwidth}
  \end{center}
  \caption{A creation event: a \+\m pair is created from a \+ defect, which is pushed one step.}\label{f:create}
\end{figure}

\section{The Geometry of Pair-Defect Collisions}\label{s:pdcollisions}

In this section we consider the collisions between a pair and a
defect. The discussion is really in two parts: On one hand, we must
consider the probability that a collision between a pair and a defect
is initiated. This depends on the density of the defects, and hence on
$\epsilon $. But, once a collision is initiated, we can ask what the
effect of the collision is going to be. The next proposition shows
that this effect is purely geometrical and independent of $\epsilon $.

\begin{proposition}\label{p:pairdefect}
There are 9 topologically different possibilities $Q_i, i=1,\dots,9$ for a collision to
be initiated. For each of them,
  there are 2 purely geometrical constants $P_{{\rm m},i}>0$ and $P_{{\rm
    d},i}>0$ (depending on $i$) which
  tell us the probability that a collision leads to a move ($P_{{\rm
    m},i}$) of a defect (by 1 or 2 links) resp.~the deletion of the pair
  ($P_{{\rm d},i}$).
\end{proposition}

The remainder of the section deals with the proof of \pref{p:pairdefect}.

\subsection{Definition}

We will study in detail how collisions move defects.
First of all, we will define what we mean by a collision.
Assuming that the density of defects is very small, the only
collisions we will consider are those involving $3$ defects, namely,
the pair \+\m which will collide with a defect \+ or \m.
\begin{definition}\label{def:collision}
 Consider some configuration $T$. Three defects $D_i, i=1,2,3$ of $T$ are said to be \emph{in a collision}
if there are $k\ge2$ flips ($k$ links of $T$) that do not increase the energy such that
\begin{enumerate}
 \item The only defects involved in these $k$ flips are $D_i,i=1,2,3$.
 \item All $3$ defects are involved in these $k$ flips.
 \item At least one of these $k$ flips will move a pair (the others might be any of the 4 cases which do not increase the energy).
\end{enumerate}
\end{definition}

\subsection{Collision types}
In this section, we will describe all possible configurations of a collision and we will show 
that the probability of any such configuration depends solely on the topology (and not on the temperature).
 
The third condition of the definition of a collision states that we can always identify a pair; as a result, the set of all possible configurations of a 
collision can be obtained by taking a pair and placing either a \m defect or a \+  defect in any
position where it can interact with one of the pair's 2 defects.
As seen in the previous section, a \+  defect can interact with any defect if and only if both defects are at distance one. Two \m defects
can interact if and only if they are on opposite corners of 2 adjacent
triangles.
The last ingredient is that \+ defects can have a degree of 6 or 8 whereas \m defects have a degree of 4 or 6.
This yields a systematic method of constructing all possible configurations of a collision:
consider a pair and let $\mathcal{U}_1$ be the set of all empty sites
(charge $0$) which are at distance 1 of any of the pair's 2 defects and $\mathcal{U}_2$
be the set of all empty sites which are opposite to the \m defect of the pair.
The set of all possible configurations of a collision is obtained by placing a \+  defect in any of $\mathcal{U}_1$'s sites or a 
\m defect in any of $\mathcal{U}_2$'s sites, as shown in \fref{f:all-collision-types} in the case where the \+ defect is red and 
the \m defect is blue.
All in all, we get 9 different configurations of a pair and a defect (symmetrical case omitted).
\begin{figure}
  \begin{center}
    \epsfig{file=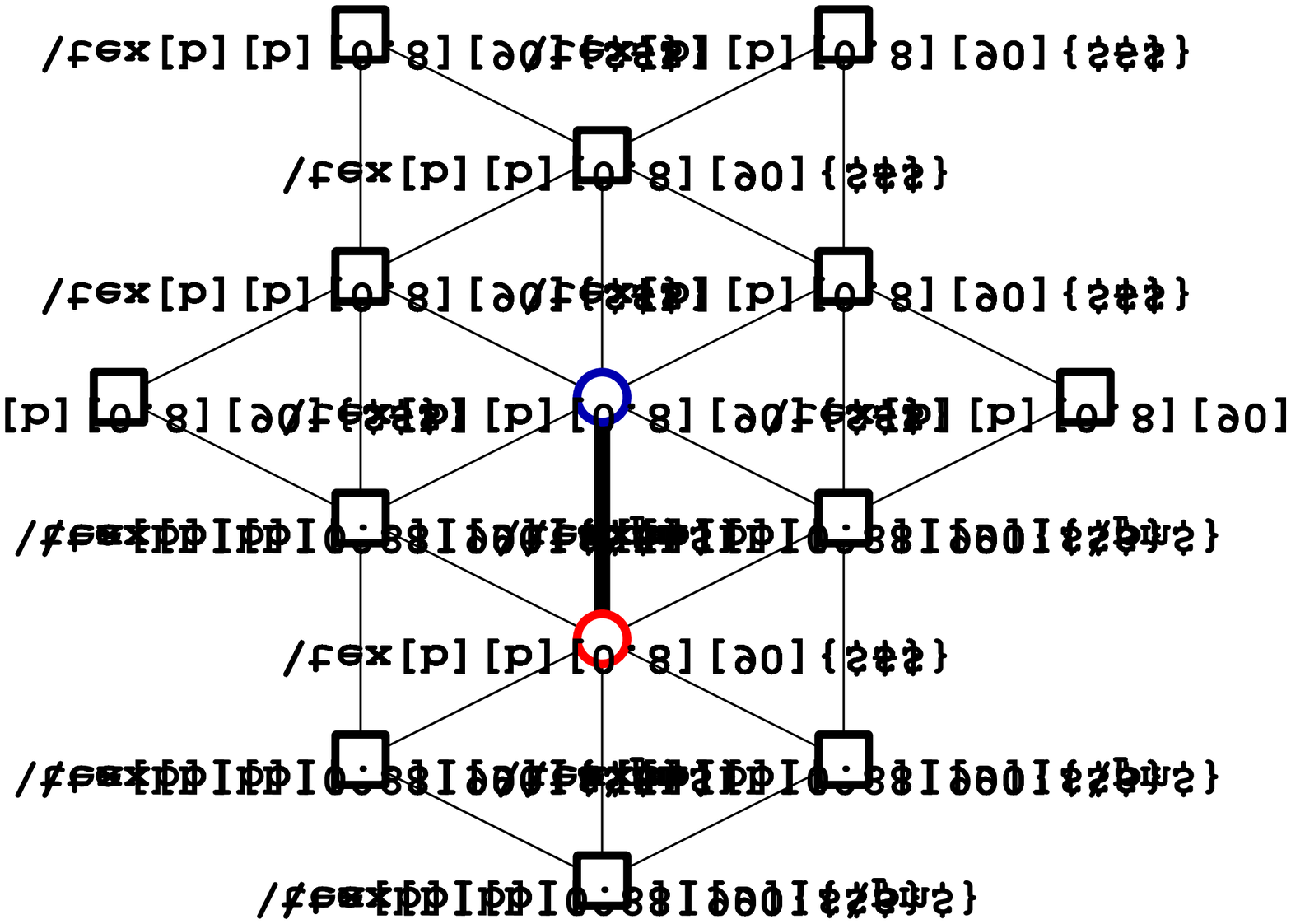,width=0.5\textwidth,angle=270}
  \end{center}
  \caption{A figure showing all possible relative positions of a pair
    and a defect in collision. The pair is shown as a solid black line.}\label{f:all-collision-types}
\end{figure}

Assuming that the defects are randomly distributed, it is clear from \fref{f:all-collision-types}
that the probability of a collision to be of some type $Q_i\in\{1,\dots,9\}$
is a temperature independent constant that can be calculated.
To prove \pref{p:pairdefect}
one must study in detail each of the 9 cases. We will study in particular:
\begin{itemize}
 \item What are the possible outcomes of each collision type and what
   is the (conditional on having initiated the collision $Q_i$) probability of each outcome?
 \item What is the probability (conditional on having initiated the collision $Q_i$) that a pair pushes a defect?
\end{itemize}
We can summarize the answers as follows:
\begin{itemize}
 \item There will always be a defect left over at the end of the collision.
 \item Finding a pair and a defect at the end of the collision is possible in all 9 cases.
 \item An annihilation of the pair is possible in 2 of the 9 cases.
 \item It is possible that the defect is pushed in 8 cases. A defect can be pushed by more than 1 step.
 \item It is possible that the defect remains in its initial position in all 9 cases.
\end{itemize}
The relative probabilities of any of the above  outcomes only depend on the
local geometry. While all the cases have been worked out in detail, we
illustrate the discussion for just 2 of them, and this will complete
the proof of \pref{p:pairdefect}.

\subsubsection{Example 1: A possible annihilation}
There are 2 cases where an annihilation might occur. We consider here
the case of \fref{f:collision-t1}. A \+\m pair collides with a \m
defect. For simplicity, assume that the \+\m pair came from the left.
Once the pair and the defect are in collision, there are 3 links whose
flipping leads to $\dE\le0$. Two of these links (the green ones)
allow the pair to walk away from (or enter) the collision. Flipping the red link on the other
hand causes an annihilation: the pair is destroyed and the defect is pushed by one step.
We clearly see that there are 3 possible outcomes:
\begin{itemize}
 \item The pair exits the collision through the same way it entered
   (in our case, on the left). The defect remains in its initial position.
 \item The pair exits the collision through the other green link. The defect moves 2 steps.
 \item An annihilation event occurs. The pair disappears and the defect moves 1 step.
\end{itemize}
The (conditional) probability of each outcome
is 1/3 and the (conditional) probability that the defect will have moved at the end of the collision
is 2/3.

\begin{figure}
  \begin{center}
    \epsfig{file=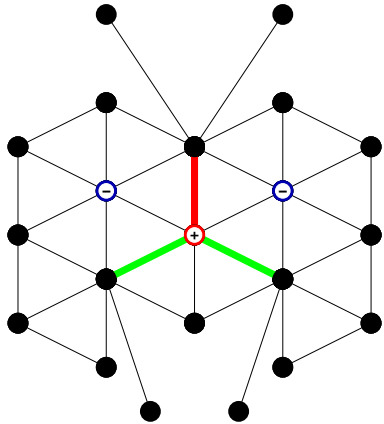,width=0.5\textwidth}
  \end{center}
  \caption{A collision where a an annihilation is possible.
  The green links show the way the pair enters (or exits) the collision.
  Flipping the red link will cause an annihilation event.}\label{f:collision-t1}
\end{figure}

\subsubsection{Example 2: A bifurcation}
Here, we look at the collision case of \fref{f:collision-t4}.
\begin{figure}
  \begin{center}
    \epsfig{file=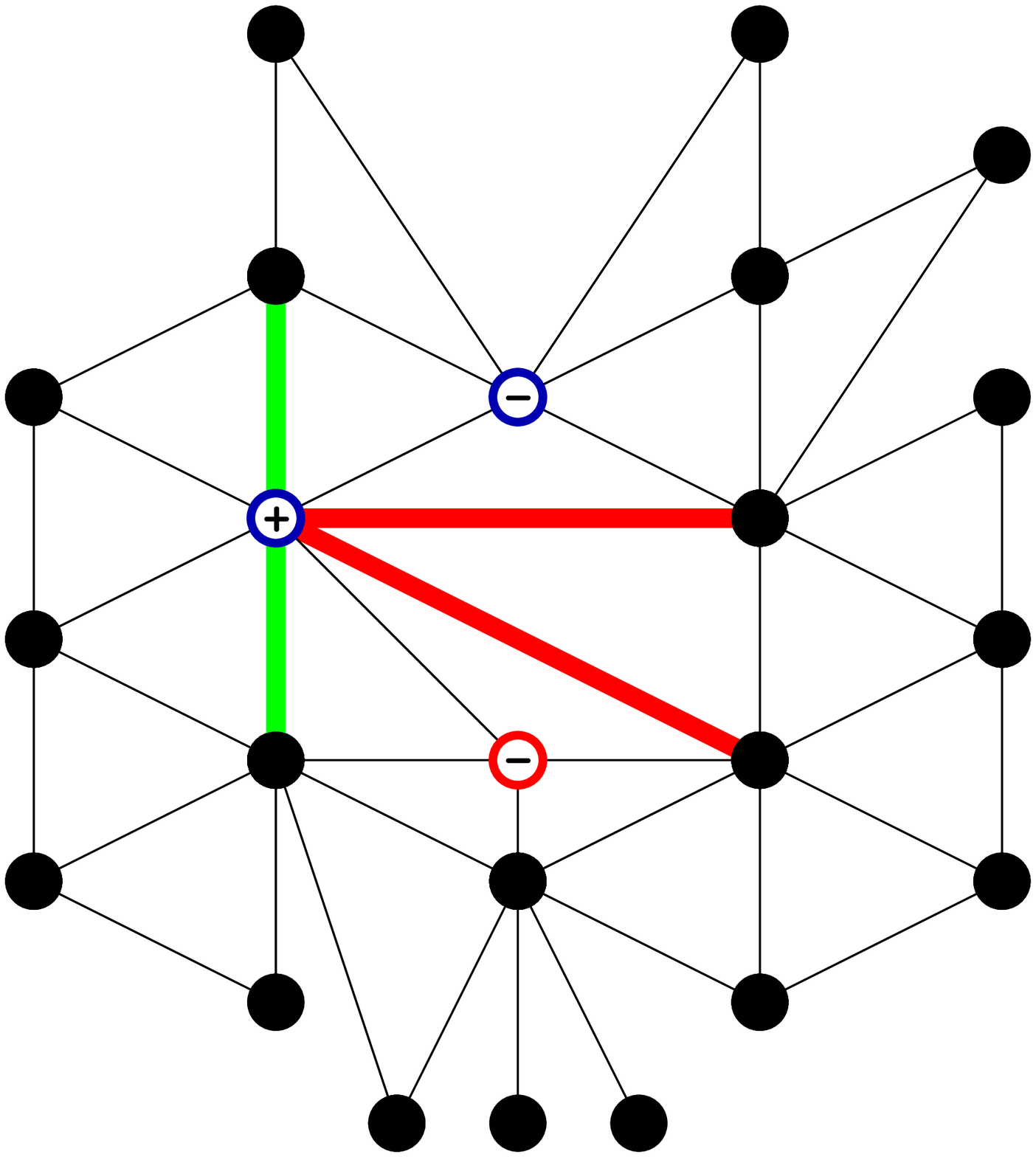,width=0.4\textwidth}\quad \raisebox{80pt}{\epsfig{file=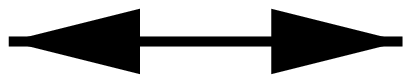,width=0.07\textwidth}}\qquad\quad
    \epsfig{file=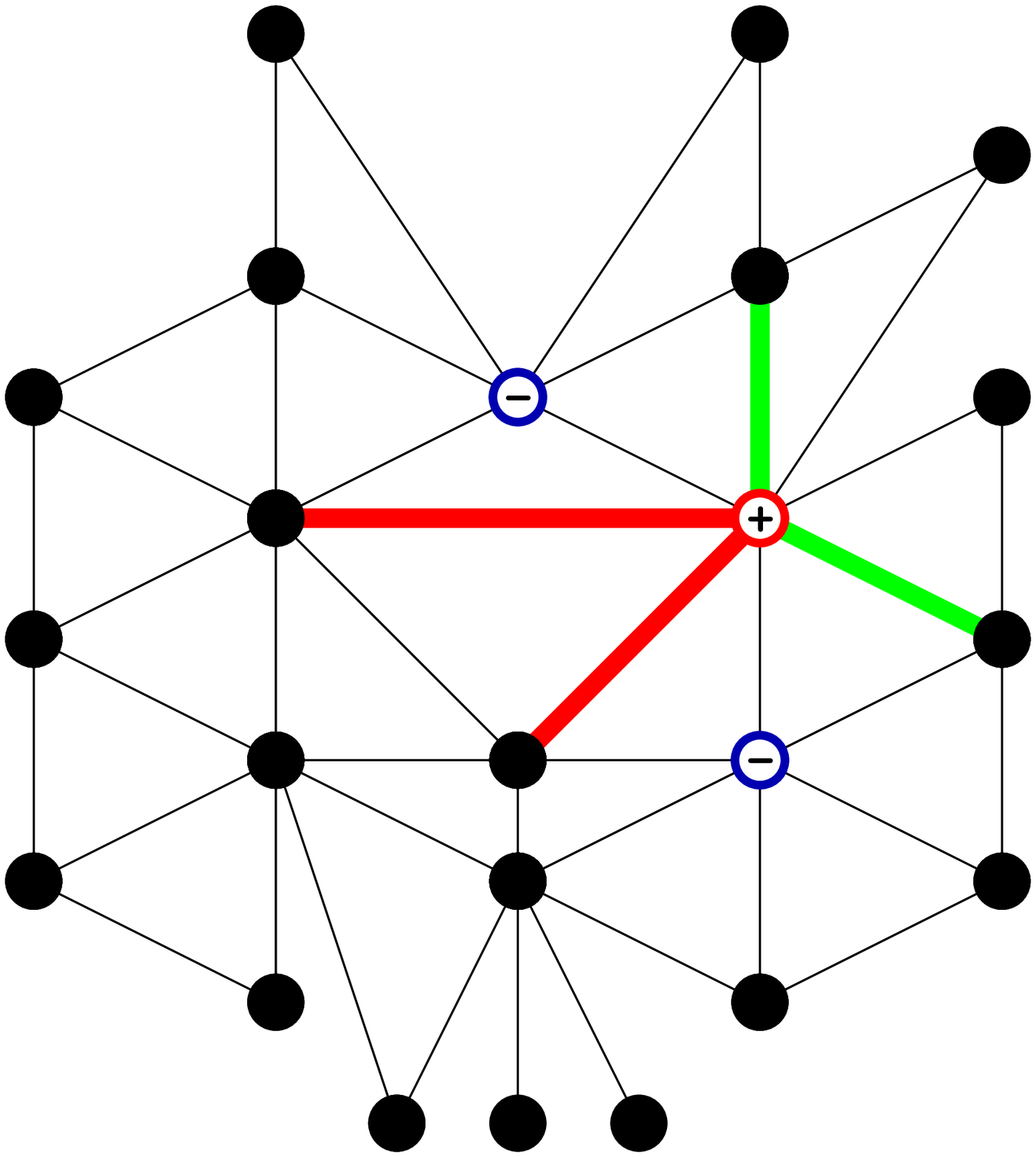,width=0.4\textwidth}\\
\hspace{0.3\textwidth}\epsfig{file=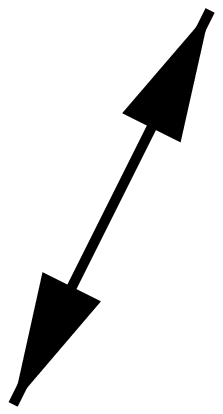,width=0.04\textwidth}\\
\raisebox{20pt}{\epsfig{file=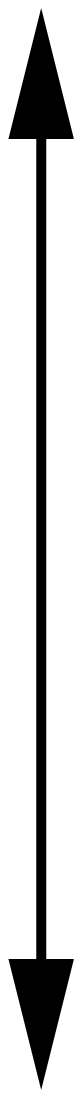,width=0.04\textwidth}}\hspace{0.2\textwidth}    \epsfig{file=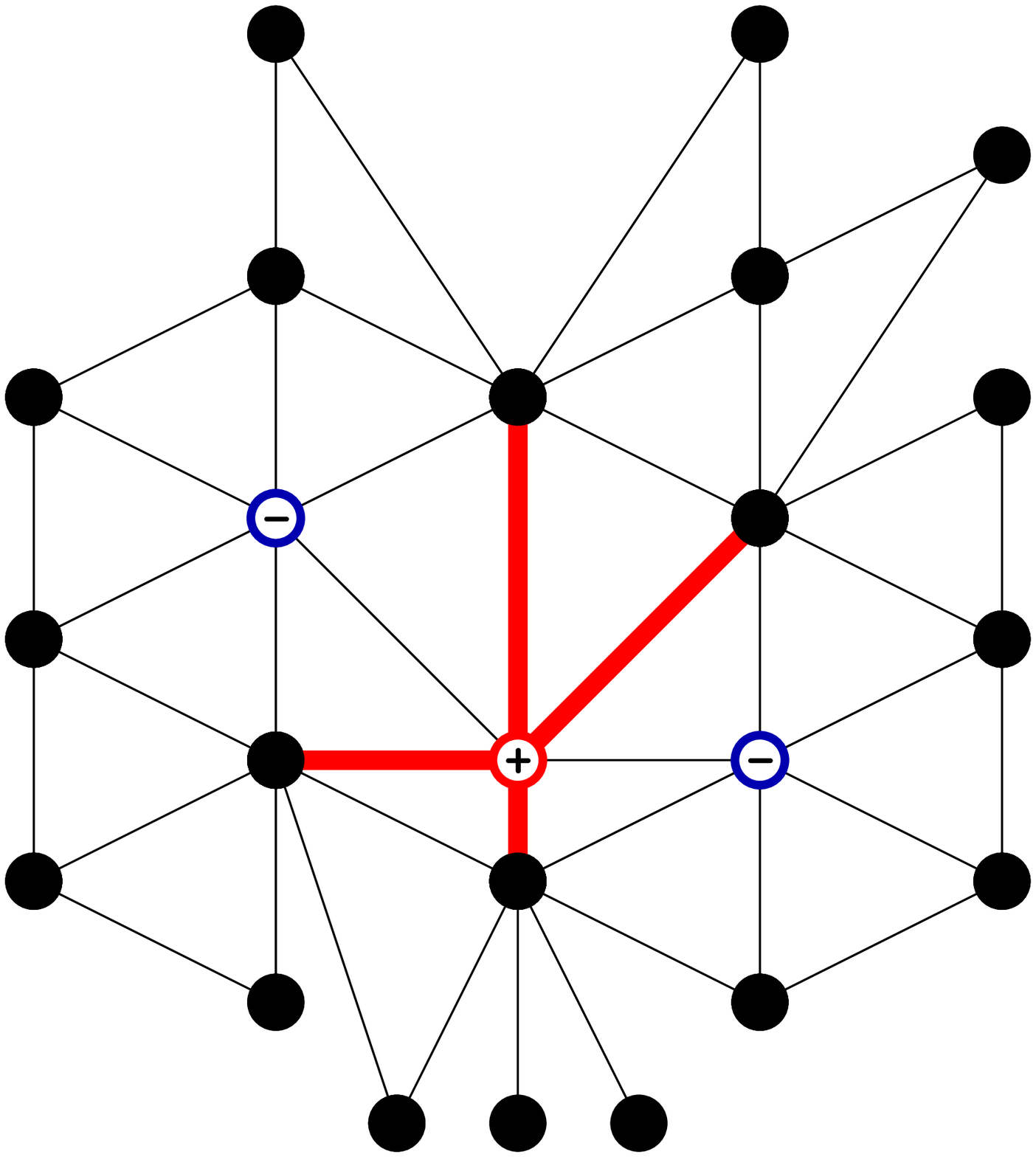,width=0.4\textwidth}\\
\hspace{0.3\textwidth}\epsfig{file=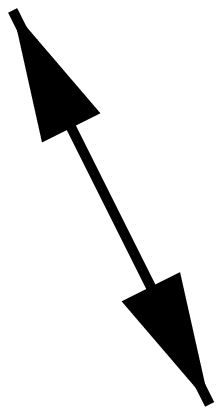,width=0.04\textwidth}\\
    \epsfig{file=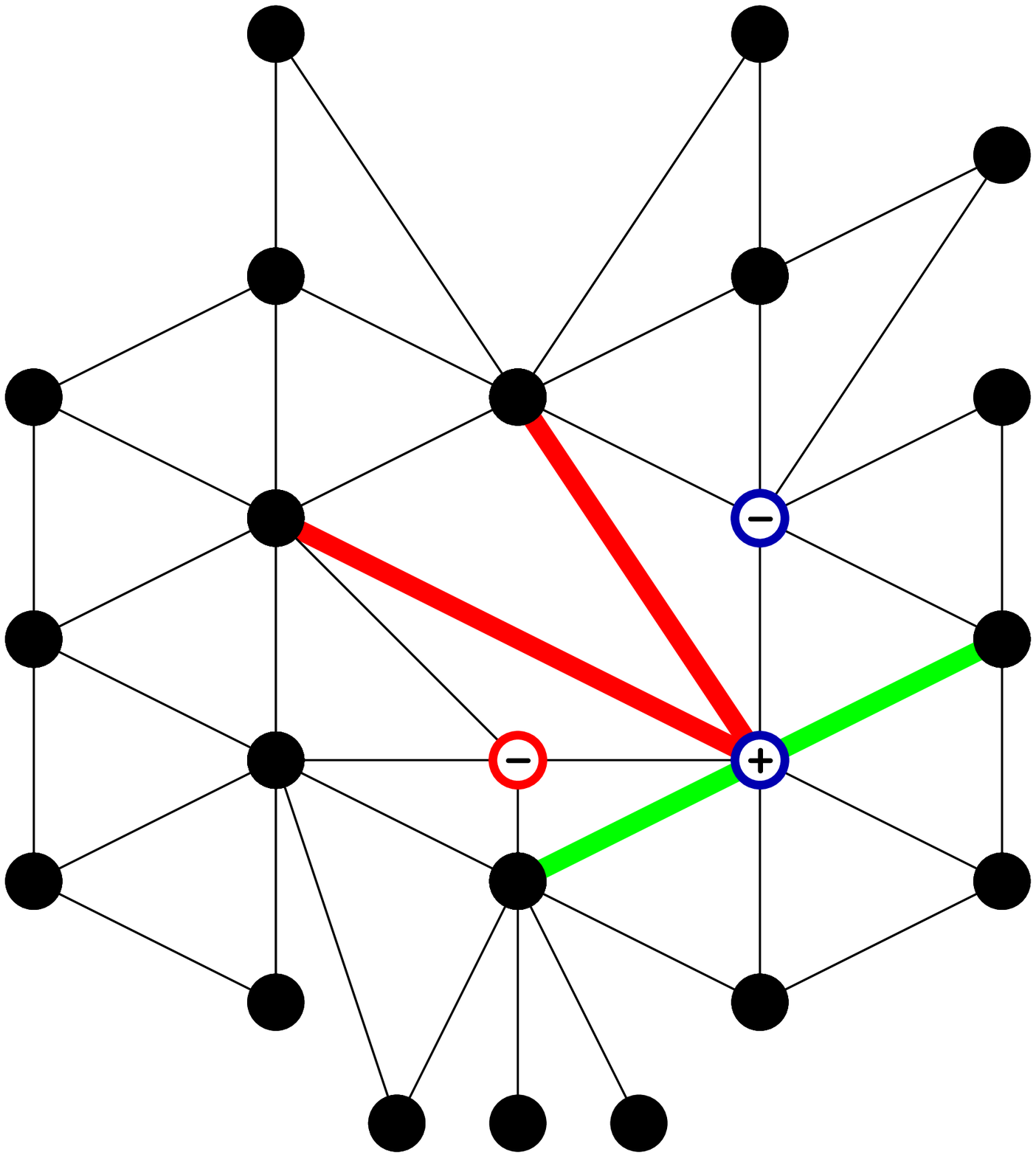,width=0.4\textwidth}\quad \raisebox{80pt}{\epsfig{file=arrow_horiz.eps,width=0.07\textwidth}}\qquad\quad
    \epsfig{file=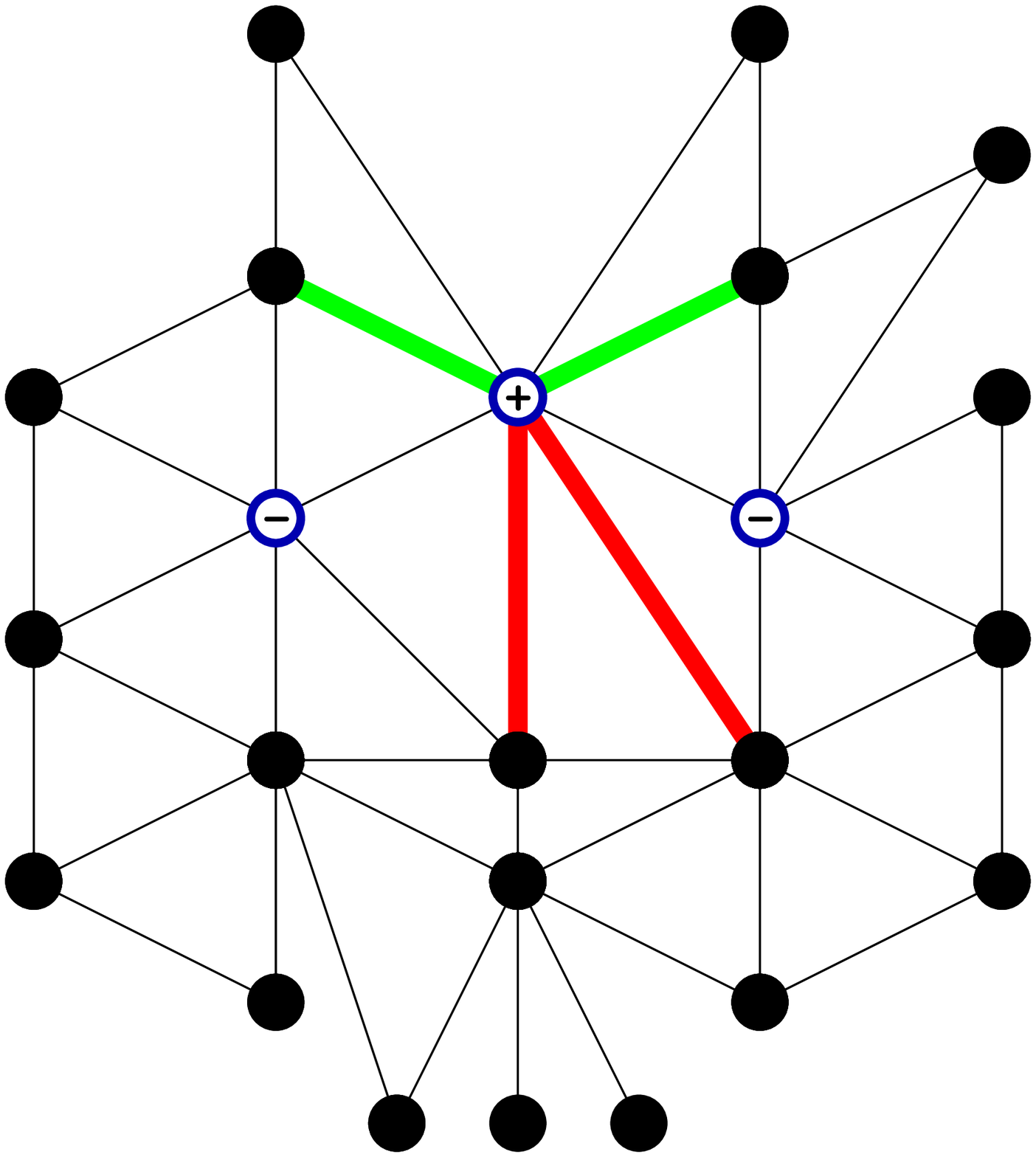,width=0.4\textwidth}
  \end{center}
  \caption{The central figure (with only red links) is symmetric along the axis {\tt-}\+\m.
 If we flip the long vertical line, we arrive at the figure
 top-right. If we flip in it the red link which does not lead back to the
 center, we arrive at top-left. Flipping the red link which does not
 lead back to top-right, we arrive at bottom-left, then at
 bottom-right, and then back to the center.
Since the same happens for the two lower links of the center, we see
that the local state space is a figure ``8'' with 9 nodes of which 8
have two exits each. The state space can be symbolized as in
\fref{f:collision-d4}.}
\label{f:collision-t4}
  \end{figure}
  No annihilation is possible here and the outcome of the collision is always one pair and one defect.
  The only relevant question is what is the probability that the
  defect will have moved at the end of the collision. But the
  combinatorics is more involved.
  
  The pair enters and may exit the collision through a green link.
  Flipping a red link on the other hand will not end the collision. Notice that the fifth diagram contains 4 red links and no 
  green ones. Moving a red link will visit the 6 figures
  sequentially. But moving the  two lower red links in the lowest
  left figure will lead to another circle of five configurations,
  which is not shown in the figure.
  This collision case can be represented by a ``state diagram'' as in
  \fref{f:collision-d4}, where each node represents a state and each
  link represents the effect of flipping one of the colored links in
  \fref{f:collision-t4}.
\begin{figure}
  \begin{center}
    \epsfig{file=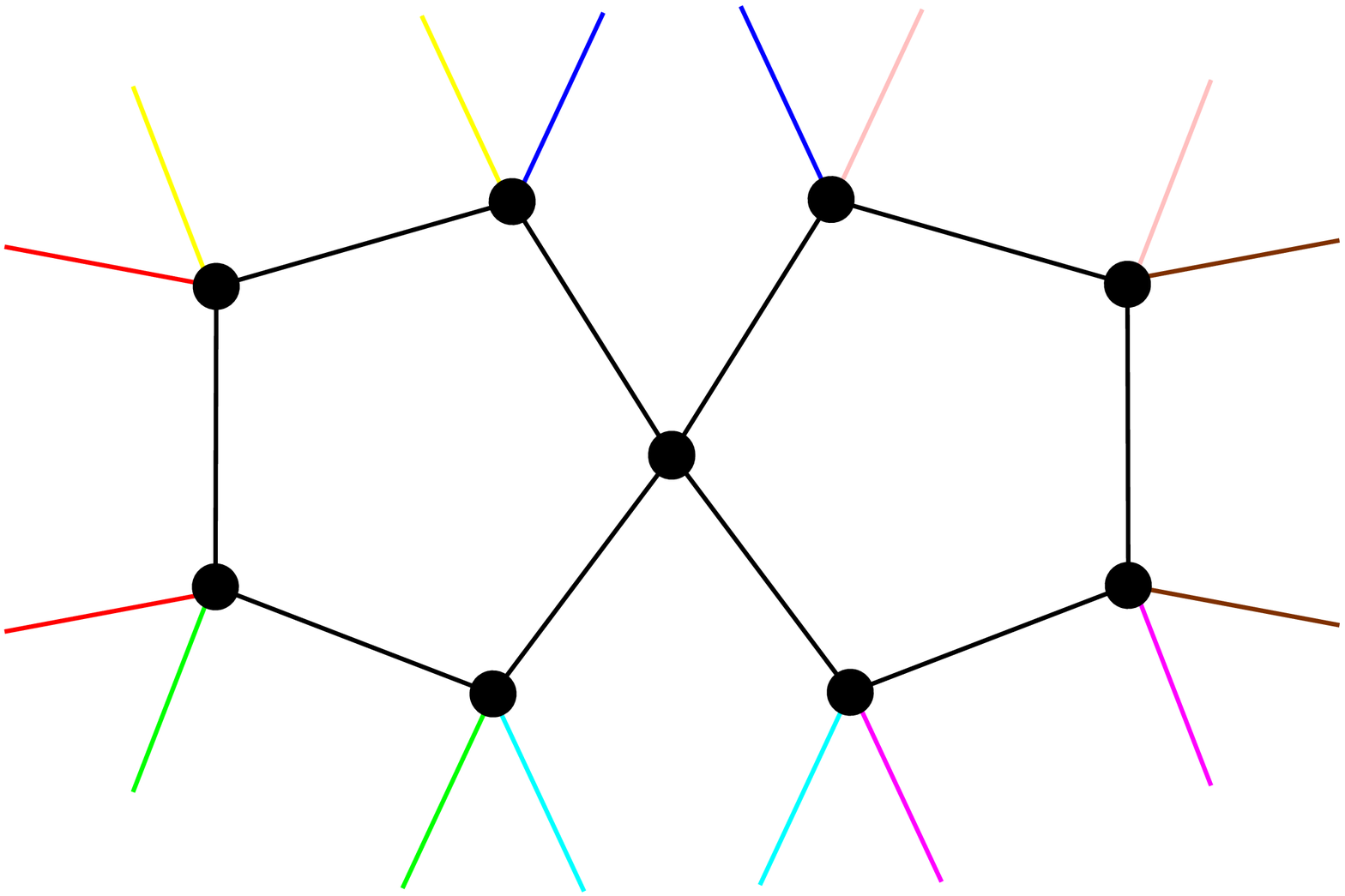,width=0.5\textwidth}
  \end{center}
  \caption{Each vertex represents one possible configuration during
    the collision of \fref{f:collision-t4}. 
  Two vertices are linked if one can go from one configuration to the
  other by flipping a (red) link of \fref{f:collision-t4}.
  The pair enters and exits the collision through one of the 16 dangling links. 
  If these 2 dangling links are the same or if they are of the same
  color, then the pair does not push the defect, otherwise, it is
  pushed by 1 step. }\label{f:collision-d4}
\end{figure}
The pair enters the collision through a dangling link $\ell_1$. It can
wander around the vertices of the state diagram before exiting through a dangling link $\ell_2$.

If $\ell_1 = \ell_2$, then it is as if the collision never occurred. In particular, the defect does not move.
Furthermore, if $\ell_1$ and $\ell_2$ are of the same color, then the defect will remain in its initial position at the end of the collision.
Using this remark and the diagrams of \fref{f:collision-t4}, one can explicitly compute the (conditional) probability that a pair pushes a defect if the collision is of the above type.
This probability will be temperature independent.

The other 7 cases are treated similarly, and this completes the proof
of \pref{p:pairdefect}.

Note that the proof means that collisions lead, on average to a
\emph{positive} probability of moving a defect. \emph{This mechanism
  is the basic reason for the diffusive wandering of the defects in
  the triangulations. It is mediated by the collision of pairs with
  the defects. Clearly, if there are no pairs, the defects can not
  move by
  this mechanism, but only through much less probable events.}

\section{Relevant and Irrelevant Pairs}\label{s:irrel}

In \sref{s:lifetime}, we have seen that a pair lives long enough to
explore its 1D path, before being destroyed by other mechanisms. We
now analyze in detail what can happen during this exploration phase.

When a pair is
created, it is one step away from its birthplace. It
will then perform a random walk on its predefined 1-dimensional path.
Each time it comes back to its birthplace, it can die with
probability $p_{\rm death}=1/3$ as shown in \fref{f:collision-t1}. \emph{If this
  happens, the triangulation will not have changed.} We will call this
an ineffective pair. The probability $P_{\rm I}=P_{\rm I}(\xi)$ can be
estimated as follows:

Assume that a defect $X$ is at a distance $\xi$ from the birthplace
of the pair. Then, by extending slightly the gambler's ruin principle \cite{Feller1957},
\emph{the probability $P_{\rm R}=P_{\rm R}(\xi )$ that the pair actually can reach $X$ is
  $(1+(\xi-1)\cdot p_{\rm death} )^{-1}=\OO(1/\xi)$}. This implies that the probability for any event implying
$X$ when starting from the birthplace depends on $\xi$, and in the
case of many defects, on their average distance (which we call again $\xi$).
Thus,
\begin{equ}\label{e:PI}
P_{\rm I} = 1-\OO(1/\xi)~, \quad P_{\rm R}=\OO(1/\xi)~.
\end{equ}

\section{Time Correlations at Equilibrium}\label{s:timecorr}

Here, we estimate the rate of change of triangulations (as a function
of time).
Since our triangulations are purely topological, we need to define
what we mean by the distance between 2 triangulations $T_1$ and $T_2$
in $\TT_N$ (the space of triangulations of the sphere with $N$ labeled nodes). There are many possible choices, see \eg,
\cite{davisonsherrington2000} many of which lead to equivalent
metrics. The one
defined below
is convenient for our purpose.

Let $\left\{T_1,T_2\right\}\subset\TT_{N}$.
Consider a node $n$ of $T_1$. The flower $f\left(n,T_1\right)$ of $n$ is defined
as the ordered cyclic set of all neighbors of $n$ in $T_1$.
Two flowers are then said to be equal if one can be obtained from the other by a cyclic rotation. 
We can now define the following metric on $\TT_{N}$:
\begin{equa}
 d\left(T_1,T_2\right)&=\sum_{n=1}^{N}{d_n\left(T_1,T_2\right)}~ \text{and} \\
 d_n\left(T_1,T_2\right)&= \left\lbrace 
  \begin{tabular}{ l l }
    0 & if $f\left(n,T_1\right)=f\left(n,T_2\right)$~, \\
    1 & otherwise~.\\
  \end{tabular}
\right. 
\end{equa}
Using this metric, we define the following correlation function:
\begin{equ}
 C(\theta)=1-\frac{d\left(T(t),T(t+\theta)\right)}{N}\text{~,}
\end{equ}
where $T(t)$ is the system state at time $t$.
Our result for the decay of this function at equilibrium, \ie, when
$t\rightarrow\infty$, is as follows:
\begin{proposition}\label{decay}
 The correlation function $C$ decays like
\begin{equ}\label{e:cor}
 C(\theta)=e^{-{\theta}/{\tau_r}}\text{~,}
\end{equ}
with a relaxation time $\tau_r$ of the form
\begin{equ}\label{e:taur}
 \tau_r=\OO(e^{3\beta })=\OO(\epsilon^{-3})\text{~.}
\end{equ}
\end{proposition}
\begin{proof}
The correlation function $C(\theta)$ is nothing but the fraction of nodes whose flower
is unchanged after $\theta$ time units.
At equilibrium, the number of pairs $p$ was established in
\eref{e:qpm} to be
$p=70/36\cdot \epsilon^2$. On the other hand, the density of defects
in equilibrium is $\OO(\epsilon )$ and hence, their average distance
$\xi$ equals $\xi=\OO(\epsilon ^{-1/2})$. By the estimates of
\sref{s:irrel}
 this means that the effective number of pairs which
change the configuration in a permanent way is $\OO(p\cdot\epsilon
^{-1/2})$. 
\begin{figure}
 \begin{center}
    \epsfig{file=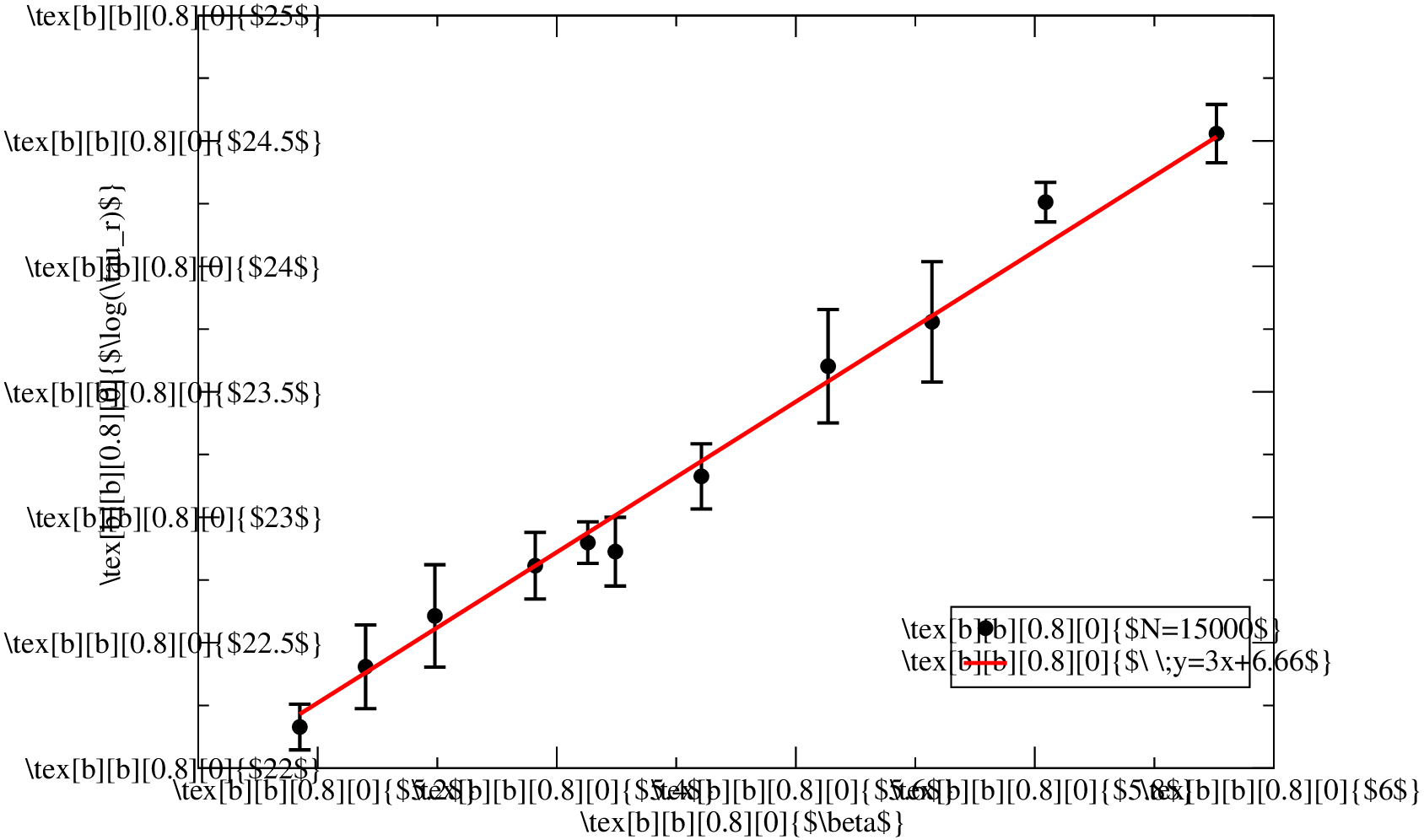,width=0.5\textwidth,angle=270}
  \end{center}  
     \caption{Decay rate of correlations at equilibrium. Numerical
       verification of Eqs.~(\ref{e:cor}) and (\ref{e:taur}). The data
     are averages over 10 runs with $N=15'000$. The error bars
     represent 1 standard deviation. The variable $\beta $ is equal to
     $-\log(\epsilon )$. The fits are for $C$ between $0.5$ and $0.001$.}\label{f:correlation}
  \end{figure}
We further saw in \sref{s:pdcollisions} that the number of collisions 
a relevant
pair will undergo is a temperature independent constant $\nu=\OO(1)$.
If $\xi$ is the average distance between 2 defects,
then, on average, this pair will change, on its way, the flowers of
$2\nu\xi$ nodes. 
At time $\theta$, each of these flowers is still unchanged with probability $C(\theta)$.

Since the pair makes a 1D random walk, all this happens within
an average time interval $\delta \theta =
\frac{1}{2}\nu^2\xi^2$. 
This in turn leads to
\begin{equ}
 C(\theta +\delta \theta ) = C(\theta)-2pP_{\rm R}\nu \xi C(\theta)\text{~.}
\end{equ}
In the limit $\theta \gg \delta \theta$, we find
\begin{equ}
 \dot{C}(\theta)=-4\frac{pP_{\rm R}}{\nu \xi}C(\theta)\text{~,} 
\end{equ}
and this leads to \eref{e:cor} with
\begin{equ}\label{e:taur2}
 \tau_r = \frac{\nu \xi(\epsilon)}{4p(\epsilon)P_{\rm R}(\xi)}\text{~.}
\end{equ}
Finally, using
\begin{equ}\label{e:s}
 P_{\rm R}(\xi) = \OO(1/\xi) = \OO\left(\epsilon^{1/2}\right)\text{~,}
\end{equ}
\eref{e:taur} follows from Equations~(\ref{e:taur2}) and (\ref{e:s}).
\end{proof}

\section{The Aging Process}

By the aging process, we mean the approach of the energy to its
equilibrium value. Since the energy is by and large just the density
$d(t)$ of defects we can formulate the result as
\begin{estimate}
Under the assumptions $\epsilon N>D_0$ and $\epsilon <\rho$ one has
for the density $d$ of defects:
 \begin{equ}\label{e:bound}
 d(t) = \OO(\left( \epsilon^2t \right)^{-2/5})\text{~.}
\end{equ}
 \end{estimate}
Note that this result differs from that proposed in \cite{sh2002}, where
the decay rate was given as $(\epsilon ^2t)^{-1/2}$. This difference
is caused by our observation that the diffusion constant of the defects
actually depends on their density, because, if they are rarer, the
pairs, which are the only ones able to move them around, need longer to
find them.

Power decay rates are extremely hard to distinguish, but we have
performed some tests which are illustrated in \fref{f:decay}. They give
a slight advantage to a decay of $-0.4$ as compared to $-0.5$. 
\begin{figure}
   \begin{center}
    \epsfig{file=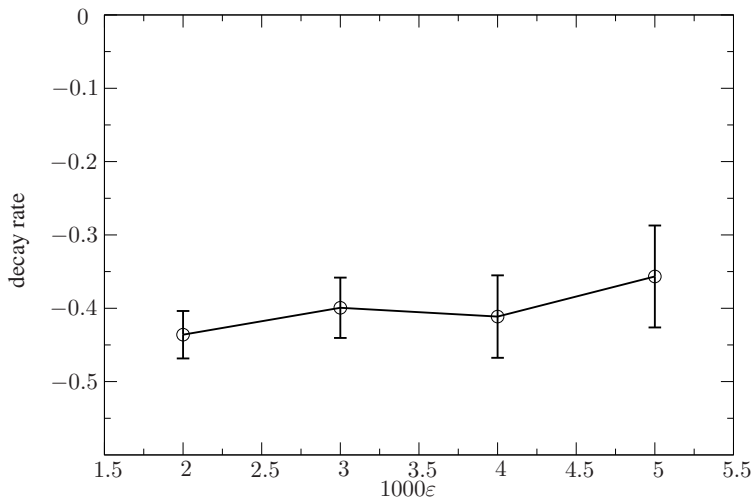,width=0.5\textwidth,angle=270}
  \end{center}
   \caption{The decay rates of several simulations with $N=15000$ and
    $\epsilon=0.002$ to $0.005$.}\label{f:decay}
  
\end{figure}

\begin{proof}We study the aging process by assuming that, in approach to equilibrium, the system
is in a quasistationary state, with charge density $c=E/N$.
Here, and
in the sequel, time will be in units of $\tau=(3N-6)/2$. 
Let $d(t)$ and $p(t)$ be the density of defects and pairs respectively.
Then, up to terms of order $\OO(\epsilon ^3)$ one has $c=d+2p$.

As we will see in this section, the quasistationarity assumption simply means 
that the relaxation of the energy is a consequence
of the \emph{annihilation of colliding defects}. The number of
pairs is, up to fluctuations, essentially unchanged during the process
we consider.

\subsection{Three timescales}
We saw that a fraction $1-\OO\left( \epsilon \right)$  of all
occurring flips in the system do not change the energy, and are either
motions of pairs or blinkers.
Of those, the only relevant ones are  the wandering pairs,
which induce diffusion of the defects as we have seen in
\sref{s:timecorr}.
The discussion of the
equilibrium probabilities apply also to states close to equilibrium,
which is the regime we want to consider now.

The pair dynamics happens on the time scale $\tau_{\rm pair}=\tau$ and it conserves both the 
number of pairs $p(t)$ and the number of defects $d(t)$.

The next slower time scale concerns creation and annihilation of
pairs. 
Even though this
changes $p(t)$, it conserves $d(t)$.
Whenever one of these events happens, defects are pushed around
by the pairs with some geometrically defined probability, and this
leads to a diffusion, whose constant $D(t)$  measures this second time
scale $\tau_{\rm diffusion}=D^{-1}(t)$.

The third time scale $\tau_{\rm meeting}$ is related to collision
rate $\gamma(t)$ of defects; $\tau_{\rm meeting}=\gamma^{-1}(t)$. They
undergo a 2D random walk. Sooner or later, 2 defects of
opposite charges  
will meet and will form a new pair which will run on timescale $\tau $
until it annihilates. 
In the regime we consider, \emph{only}
this sequence of events (collision and running pair) of the dynamics destroys 2 defects and, as a consequence, is responsible for the relaxation of
the energy. Given the 3 time scales, the derivation of the decay rate
is now rather straightforward.

\subsection{The quasistationarity assumption and the density of pairs}
By the previous discussion,
\begin{equ}
   \tau_{\rm meeting}(t) \gg  \tau_{\rm diffusion}(t)\gg \tau_{\rm pair}(t)=\tau=1~.
\end{equ}
The orders of magnitude of these quantities near equilibrium are
\begin{equ}
  \tau_{\rm meeting}(t) =\OO(\epsilon^{-2} d^{-7/2})~,\quad \tau_{\rm
    diffusion}(t)=\OO(\epsilon ^{-2}d^{-1/2}), \quad \tau_{\rm pair}(t)=1~.
\end{equ}
Consider a system for which, at time 0, $d(0)\gg1$ and $p(0)\gg1$.
It is clear that the relaxation of pairs is much faster than that of defects.
We will assume that pairs are always at equilibrium density, \ie, that
creation and destruction rates of pairs are equal and $p(t)$ is
independent of $t$.
\begin{remark}
 The above discussion implies that $p(t)$ is 
 constant over time intervals of order $\tau_{\rm meeting}(t)$. In
 fact, both creation and annihilation 
 events necessitate the presence of defects so that the creation and destruction rates of pairs 
 will be linear in $d(t)$ at low density. This implies that $p$
 depends on $t$ only through the value of $d(t)$.
By abuse of notation, we will write $p(d)$ instead of $p\bigl(d(t)\bigr)$.
\end{remark}

The creation rate of pairs is $12d\epsilon^2$ and the destruction rate is simply $p(d)/\tau_{\rm lifetime}$.
Therefore, by balancing the rates, we find:
\begin{equ}\label{e:p}
 p(d)=12d\,\tau_{\rm lifetime}\,\epsilon^2\text{~.}
\end{equ}
Since a pair needs to diffuse from one defect to the other in order to annihilate, we estimate 
that $\tau_{\rm lifetime}=\OO\left( \xi^{2} \right)=\OO( d^{-1})$.
This implies that the density of pairs 
is
$p(d)=\OO( \epsilon^2)$.

\subsection{The diffusion constant of single defects}
Repeating the arguments of \sref{s:timecorr}
the average number of collisions $\nu$ and the average number
of moved defects $\eta$ are temperature independent constants.
The diffusion constant of a defect is simply the probability that a
given defect moves by one space unit during one time unit and it is given
by 
\begin{equs}
 D\cdot d &= \frac{2p(d)\,P_{\rm R}(\xi)\eta}{\nu^2\xi^2}\text{~,} \\
 D(d) &=\OO\bigl( p(d) \cdot P_{\rm R}(\OO(d^{1/2}))\bigr)\text{~.}
\end{equs}
Using Equations~(\ref{e:p}) and (\ref{e:s}), this leads to
\begin{equ}
 D(d) =\OO( \epsilon^2 \cdot d^{1/2})\text{~.}
\end{equ}

\subsection{Collision rate of single defects and relaxation coefficient}
The annihilation of 2 diffusive particles $A+B\rightarrow\emptyset$
has been studied in depth in
\cite{ToussaintWilczek1983,Cardy1995,Sasada2010}. Here, we use the
mean field argument of \cite{ToussaintWilczek1983}, to deduce the
collision rates. However, there will also be particle creation. On the
other hand, \eg, in \cite{Sasada2010} creation is indeed considered,
but the study is for a fixed substrate, namely the lattice
$\mathbb{Z}^2$, while our study is on a more floppy domain.

Given a 2D gas of 2 particles $A$ and $B$ of equal densities $d/2$
such that the diffusion constants $D_A=D_B=D$, it can be deduced from 
\cite{ToussaintWilczek1983} that the collision rate $\gamma$ is
\begin{equ}
 \gamma(d) = \OO(  D d^3)\text{~.}
\end{equ}
Extending this identity to a varying diffusion constant,
we end up with
\begin{equ}
 \dot{d}=-2\gamma(d)=-\OO( \epsilon^2 \cdot d^{7/2})\text{~,}
\end{equ}
where we assumed that we are far enough from equilibrium to neglect the creation rate of defects.
\end{proof}

Note that this result differs from that proposed in \cite{sh2002}, where
the decay rate was given as $(\epsilon ^2t)^{-1/2}$. This difference
is caused by our observation that the diffusion constant of the defects
actually depends on their density, because, if they are rarer, the
pairs, which are the only ones able to move them around need longer to
find them.

\section*{Acknowledgements}This research was partially supported by
the Fonds National Suisse.


\def\Rom#1{\uppercase\expandafter{\romannumeral #1}}\def\u#1{{\accent"15
  #1}}\def\cprime{$'$} \def\cprime{$'$}

\end{document}